\documentclass[letterpaper,11pt]{article}
\usepackage{amsmath,amssymb,amsthm}
\usepackage{graphicx,xcolor}
\usepackage{tikz}
\usepackage{wrapfig}   % for wrapfig function
\usepackage{url}
\usepackage{multirow}
\usepackage{tcolorbox}
\tcbuselibrary{skins}
\usepackage{environ}
\usepackage{colortbl}
\usepackage{hhline}
\usepackage{rotating}
\usepackage{array}
\usepackage{enumitem}
\usepackage{mdframed}
\usepackage{cancel}
\usepackage{soul}
\usepackage{authblk}
\usepackage[margin=1in]{geometry}
\usepackage[english]{babel}
\usepackage[utf8x]{inputenc}
\usepackage[colorinlistoftodos]{todonotes}
\usepackage{balance}
\usepackage{xspace}

\usepackage{fontawesome}
\usepackage{amsthm, bm}

\newcommand{\AlmostSort}{\textsc{Almost-Sort}}
\newcommand{\EpsilonWaySort}{\textsc{$n^{\epsilon}$-Way-Sort}}
\newcommand{\MainSort}{\textsc{Full-Sort}}
\newcommand{\TunableSort}{\textsc{Sort-Adaptive}}

\newtheorem{theorem}{Theorem}[section]
\newtheorem{corollary}{Corollary}[theorem]
\newtheorem{lemma}[theorem]{Lemma}

\renewcommand{\epsilon}{\varepsilon}
\newcommand{\band}{{stage}\xspace}
\newcommand{\bands}{{stages}\xspace}

\graphicspath{{images/}}

\usetikzlibrary{trees,arrows,automata}
\usetikzlibrary{shapes.geometric} % for oval nodes
\usetikzlibrary{shadows,fadings}
\usetikzlibrary{arrows,shapes,positioning,shadows,trees}

\newcommand{\hide}[1]{}

\newcommand{\ceil}[1]       {\left\lceil #1 \right\rceil}
\newcommand{\highlight}[1]{\textit{\textbf{#1}}}

\definecolor{orange}{rgb}{1,0.3,0}

\newcommand{\vgap}{\vspace{-0.2cm}}

\newcommand{\para}[1]{\vspace{-0.2cm}\noindent{\bf{#1.}}}

\newcommand{\codesize}{\small}

\newcommand{\T}{\hspace{1em}}

\definecolor{gray}{rgb}{0.3,0.3,0.3}

\newcommand{\Oh}[1]{{\mathcal O}\left({#1}\right)}

\newcommand{\oh}[1]{{o}\left({#1}\right)}
\newcommand{\Om}[1]{{\Omega}\left({#1}\right)}

\newcommand{\Th}[1]{{\Theta}\left({#1}\right)}

\newcommand{\xif}{{\bf{{if~}}}}
\newcommand{\xthen}{{\bf{{then~}}}}
\newcommand{\xelse}{{\bf{{else~}}}}

\newcommand{\xfor}{{\bf{{for~}}}}

\newcommand{\xto}{{\bf{{to~}}}}

\newcommand{\xdo}{{\bf{{do~}}}}

\newcommand{\xand}{{\bf{{and~}}}}

\newcommand{\xreturn}{{\bf{{return~}}}}

\newcommand{\xparallelfor}{{\bf{{parallel for~}}}}
\newcommand{\xpar}{{\bf{{par:~}}}}

\newcommand{\xblankpar}{{$\phantom{waa:}$}}

\newcommand{\xcomment}{\hfill $\rhd$ }

\newcommand{\vsitem}{\vspace{-0.15cm}\item}

\makeatletter
\newcommand{\ALOOP}[1]{\ALC@it\algorithmicloop\ #1%
  \begin{ALC@loop}}
\newcommand{\ENDALOOP}{\end{ALC@loop}\ALC@it\algorithmicendloop}

\makeatother

\definecolor{gray}{rgb}{0.3,0.3,0.3}
\colorlet{lightblue}{blue!15}
\colorlet{lightred}{red!20}
\colorlet{lightyellow}{green!20}
\colorlet{framecolor}{yellow!20}
\colorlet{frameloopcolor}{yellow!20}

\surroundwithmdframed[
  backgroundcolor   = algocolor,
  hidealllines      = true,
  innerleftmargin   = 0.5em,
  innerrightmargin  = 0.5em,
  innertopmargin    = -1.1em,
  innerbottommargin = 0em,
  leftline=true,
  rightline=true,
  bottomline=true,
  topline=true,
  skipabove         = .5\baselineskip,
  skipbelow         = .5\baselineskip
]{algorithm}

\makeatletter
   \newcommand\figcaption{\def\@captype{figure}\caption}
   \newcommand\tabcaption{\def\@captype{table}\caption}
\makeatother

\colorlet{algotitlebarcolor}{green!20}
\colorlet{algotitlebarcolortop}{green!25}
\colorlet{algotitlebarcolorbottom}{green!20}
\colorlet{framecolor}{yellow}
\colorlet{framecolortop}{yellow!15}
\colorlet{framecolorbottom}{yellow!15}
\colorlet{bluetop}{cyan!40}
\colorlet{bluebottom}{cyan!20}
\colorlet{greentop}{white!40}
\colorlet{greenbottom}{white!20}
\colorlet{shadowcolor}{gray}
\colorlet{algotitlecolor}{black}
\colorlet{algoframecolor}{gray}
\colorlet{tabletitlecolor}{green!30}
\colorlet{tablefillcolor}{yellow!10}
\colorlet{tablefillcolor2}{yellow!10}
\colorlet{algocolor}{black}

\newtcolorbox{mycolorbox}[1]
{skin=enhanced,colbacktitle=algotitlebarcolor,coltitle=algotitlecolor,colback=algocolor,colframe=algoframecolor,boxrule=1pt,left=5mm,right=1mm,title style={top color=algotitlebarcolortop, bottom color=algotitlebarcolorbottom},interior style={top color=framecolortop, bottom color=framecolorbottom},drop shadow,title={\codesize #1}}

\newcommand{\algotopspace}{\vspace{-0.2cm}}
\newcommand{\algomiddlespace}{\vspace{0.1cm}}
\newcommand{\algobottomspace}{\vspace{-0.15cm}}

\newcommand{\algoinput}{{\bf{{Input:~}}}}
\newcommand{\algooutput}{{\\\bf{{Output:~}}}}
\newcommand{\algorequire}{{\bf{{Require:~}}}}

\arrayrulecolor{gray}
\newsavebox{\tablebox}
\NewEnviron{colortabular}[1]{%
  \addtolength{\extrarowheight}{1ex}%
  \savebox{\tablebox}{%
    \begin{tabular}{#1}%
	  \rowcolor{tabletitlecolor}
      \BODY%
    \end{tabular}}
  \begin{tikzpicture}
    \begin{scope}
      \clip[rounded corners=1ex] (0,-\dp\tablebox) -- (\wd\tablebox,-\dp\tablebox) -- (\wd\tablebox,\ht\tablebox) -- (0,\ht\tablebox) -- cycle;
      \node at (0,-\dp\tablebox) [anchor=south west,inner sep=0pt]{\usebox{\tablebox}};
    \end{scope}
    \draw[rounded corners=1ex] (0,-\dp\tablebox) -- (\wd\tablebox,-\dp\tablebox) -- (\wd\tablebox,\ht\tablebox) -- (0,\ht\tablebox) -- cycle;
  \end{tikzpicture}
}

\setlist[enumerate,1]{leftmargin=\dimexpr 26pt-0.3cm}

\title{Low-Depth Parallel Algorithms for the\\
Binary-Forking Model without Atomics}
\date{\vspace{-5ex}}
\author[1]{Zafar Ahmad}
\author[1]{Rezaul Chowdhury}
\author[1]{Rathish Das}
\author[1]{Pramod Ganapathi}
\author[2]{Aaron Gregory}
\author[1]{Mohammad Mahdi Javanmard}
\affil[1]{Department of Computer Science, Stony Brook University}
\affil[2]{Department of Applied Mathematics \& Statistics, Stony Brook University}

\begin{document}
\maketitle

\begin{abstract}
The binary-forking model is a parallel computation model, formally defined by Blelloch et al. very recently, in which a thread can fork a concurrent child thread, recursively and asynchronously. The model incurs a cost of $\Theta(\log n)$ to spawn or synchronize $n$ tasks or threads. The binary-forking model realistically captures the performance of parallel algorithms implemented using modern multithreaded programming languages on multicore shared-memory machines. In contrast, the widely studied theoretical PRAM model does not consider the cost of spawning and synchronizing threads, and as a result, algorithms achieving optimal performance bounds in the PRAM model may not be optimal in the binary-forking model. Often, algorithms need to be redesigned to achieve optimal performance bounds in the binary-forking model and the non-constant synchronization cost makes the task challenging.

Though the binary-forking model allows the use of atomic {\em test-and-set} (TS) instructions to reduce some synchronization overhead, assuming the availability of such instructions puts a stronger requirement on the hardware and may limit the portability of the algorithms using them. In this paper, we avoid the use of locks and atomic instructions in our algorithms except possibly inside the join operation which is implemented by the runtime system. 

In this paper, we design efficient parallel algorithms in the binary-forking model without atomics for three fundamental problems: Strassen's (and Strassen-like) matrix multiplication (MM), comparison-based sorting, and the Fast Fourier Transform (FFT). All our results improve over known results for the corresponding problem in the binary-forking model both with and without atomics.

We present techniques for designing efficient algorithms without using locks and atomic instructions. We use extra space to prevent the work blow-up in the highly parallel asynchronous computations performed by our MM and FFT algorithms. Extra space also allows us to avoid the use of atomic TS instructions in our sorting algorithm as well as to achieve a stronger bound (i.e., with high probability) on its work. Though space plays a major role in the design of all our algorithms, our MM and FFT algorithms do not use asymptotically more space than their PRAM counterparts. We present space-adaptive algorithms for MM and sorting that achieve provably good performance bounds for any given amount of space. 

\end{abstract}
\clearpage
\section{Introduction}
We present efficient algorithms with optimal/near-optimal span\footnote{Span/depth is the running time of an algorithm with an unbounded number of processors.} for several fundamental problems in the binary-forking model without locks and atomic instructions. The binary-forking model was introduced in Blelloch et al. \cite{BlellochFiGuSu2019} (see also \cite{acar2000data, ben2016parallel, blelloch2008provably, blelloch2011scheduling, das2019data}) to accurately capture the performance of algorithms designed for modern multi-core shared-memory machines. In this model, the computation starts with a single thread, and as the computation progresses, threads are created dynamically and asynchronously; the computation finishes when all threads end. A thread can spawn/fork a concurrent asynchronous child thread while it progresses simultaneously and such forking of threads can happen recursively; hence the model is called the binary-forking model. The model also includes a ``join" operation to synchronize the threads. Though the model introduced in \cite{BlellochFiGuSu2019} allows the use of atomic instructions, we do not use them in this paper.

The binary-forking model is closely related to the well-studied PRAM model~\cite{jaja1997introduction}. The PRAM model is strictly more powerful than the binary-forking model; however, it does not correlate well with modern architectures. In the PRAM model, computation progresses in synchronous steps. Modern architectures employ new techniques such as use of multiple caches, processor pipelining, branch prediction, hyper-threading, and many more, which give rise to many asynchronous events such as cache misses, varying clock speed, interrupts, etc., thus demanding the development of a  parallel computation model where computation can proceed asynchronously. Asynchronous thread creation in the binary-forking model makes it an ideal candidate for modeling parallel computation in modern architectures. As pointed out in \cite{BlellochFiGuSu2019}, this is the model underlying many widely used parallel programming languages/environments such as Cilk~\cite{frigo1998implementation}, the Java fork-join framework~\cite{forkjoin}, Intel TBB~\cite{tbb}, and the Microsoft Task Parallel Library~\cite{micpparlib}. 

One can trivially reduce any algorithm designed for the PRAM model to an algorithm for the binary-forking model, incurring an $O(\log n)$-factor blow-up in the span while keeping the work\footnote{Work is the number of operations performed by a parallel algorithm on a serial computer.} asymptotically the same as in the PRAM model. Spawning $n$ threads takes $\Th{1}$ time and $\Th{n}$ work in the PRAM model---making the synchronization cost (span) constant. This synchronization can be simulated by using a binary tree of $\log n$ depth and $\Th{n}$ nodes in the binary-forking model. Each internal node in the binary tree corresponds to a binary-forking operation, and the $n$ leaves correspond to the $n$ spawned threads. 

A direct simulation of an optimal-span PRAM algorithm may not produce an algorithm with optimal span in the binary-forking model. For example, Cole's parallel merge sort~\cite{cole1988parallel} achieves optimal $\Th{\log n}$ span and $\Th{n \log n}$ work in the PRAM model. The binary-tree reduction increases the span to $\Th{\log^2 n}$ while keeping the work asymptotically the same.  On the other hand, by increasing work to $\Th{n^2}$, it becomes trivial to get a $\Th{\log n}$ span sorting algorithm --- each item independently computes its rank in the final sorted list in $\Th{\log n}$ time and $\Th{n}$ work by comparing itself with all $n$ elements. However, neither algorithm is optimal in the binary-forking model --- the former has non-optimal span while the latter performs non-optimal work. Cole and Ramachandran~\cite{ColeRa2017} presented a deterministic sorting algorithm with $\Oh{\log n\log\log n}$ span and optimal $\Th{n \log n}$ work in the binary-forking model. Recently, Ramachandran and Shi~\cite{ramachandran2020data} gave a data-oblivious sorting algorithm in the binary-forking model with optimal work and $\Oh{\log n \log \log n}$ span. Very recently, Blelloch et al.~\cite{BlellochFiGuSu2019} used atomic test-and-set operations to design a randomized sorting algorithm with $\Th{\log n}$ span w.h.p.\footnote{An event $\xi$ occurs with high probability (w.h.p.) in $n$ provided it occurs with probability at least $1 - \frac{c}{n^\alpha}$ for constants $\alpha \geq 1$ and $c > 0$.} in $n$ and $\Th{n \log n}$ work in expectation. Hence, finding an optimal (both in span and work) sorting algorithm without using atomic instructions remains an interesting and non-trivial open problem in the binary-forking model. We encounter the span blow-up problem when running other fundamental low-span PRAM algorithms, such as those for Strassen's matrix multiplication and Fast Fourier Transform (FFT), in the binary-forking model. Both algorithms have $\Th{\log n}$ span in the PRAM model, which blow up by a factor of $\log n$ and $\log\log n$, respectively, in the binary-forking model. 

The binary-forking model introduced in~\cite{BlellochFiGuSu2019} allows the  use of atomic \textit{test-and-set} (TS) operations. When performed on a shared memory location $L$, TS performs the following as a single undivided operation: it reads the value stored at $L$ and if the value is zero, sets $L$ to one and returns zero, otherwise, returns one without changing $L$. While a TS operation makes the binary-forking model arguably more powerful, it also puts a stronger requirement on the memory hardware. Hence, an algorithm may become more portable by avoiding the use of TS. Such an algorithm is also arguably more elegant~\cite{BlellochFiGuSu2019}. While several parallel algorithms have been designed for the binary-forking model without locks and atomic instructions \cite{frigo2012cache, CormenLeRiSt2009, ChowdhuryRa2006, ChowdhuryRa2008, blelloch2008provably, ChowdhuryRa2010, BlellochGiSi2010, blelloch2011scheduling, TithiMaMeCh2013, chowdhury2016autogen, ChowdhuryGaTsTiBaLeSoKuTa2017, ChowdhuryGaTaTi2017, JavanmardGaDaAhTsCh2019, das2019data,  TithiGaTaAgCh2015, Tithi2015, Ganapathi2016, TangYoKaTiGaCh2014}, developing algorithms with optimal/near-optimal span and work for several fundamental problems in the binary-forking model without locks and atomic instructions remain open.

Algorithms for the binary-forking model without atomics face two major challenges: 1) how to avoid the  blow-up in span (synchronization cost) without blowing up work? 2) how to avoid the use of atomic operations without asymptotically increasing span and work? Surprisingly, it turns out that using extra space, we can tackle both challenges. By extra space, we mean the space allocated from RAM (from heap memory), not from processors registers or stack memory. 

\paragraph{Our Contributions.} In this paper, we present results for three fundamental problems in the binary-forking model without atomics. Our major results include: \textbf{(1)} an optimal $O(\log n)$ span algorithm for Strassen's Matrix Multiplication (MM) with only a $\Th{\log\log n}$-factor blow-up in work as well as a near-optimal $O(\log n \log\log n)$ span algorithm with no asymptotic blow-up in work; \textbf{(2)} a randomized comparison-based sorting algorithm with optimal $O(\log n)$ span and $O(n\log n)$ work, both with w.h.p. in $n$; and \textbf{(3)} a near-optimal $O(\log n \log\log\log n)$ span algorithm for FFT with less than a $\log n$-factor blow-up in work for all practical values of $n$ (i.e., $n \le 10 ^{10,000}$). 

Though space played a major role in the design of all our algorithms in this paper, our algorithms for Strassen's matrix multiplication and FFT do not use asymptotically more space than their standard PRAM counterparts. We present a space-adaptive algorithm for Strassen's MM which always achieves a span within $\Th{\log n}$ factor of optimal for any given amount of space.

We list our major results in Table~\ref{tab:summary}.

\paragraph{Major Techniques.} The extra $\log{n}$ factor in the span of the standard parallelization of Strassen's MM algorithm in the binary-forking model arises from the fact that it spends $\Th{\log{{\frac{n}{2^i}}}}$ time (synchronization cost) computing intermediate results at recursion level $i$ for each $i \in [0, \log_{2}{n}]$ which requires only $\Oh{1}$ time in the PRAM model. We observe that none of those intermediate matrices need to be explicitly computed or stored to compute the final output. Indeed, each cell in the final output matrix can be computed directly in $\Oh{\log{n}}$ time from the two original input matrices of the algorithm. This \textbf{\textit{single-point computation method}} can be used to compute all the cells in the output matrix simultaneously in $\Th{\log n}$ span. However, this approach blows up the work performed by the algorithm by up to a $\Th{n^2}$  factor because the approach does not reuse intermediate results. We avoid this work blow-up by computing and temporarily storing the intermediate results at $\Th{\log\log{n}}$ carefully `chosen levels' of recursion, which eliminates the need for implicitly recomputing the intermediate results over and over again. So, all single-point computations proceed in stages where a stage includes all levels of recursion between two consecutive `chosen levels,' and synchronizations happen only at stage boundaries with all threads executing asynchronously within every stage. We show that this \highlight{stage-based approach} reduces the work blow-up from $\Th{n^2}$ factor to only $\Th{\log\log{n}}$ factor while achieving the optimal $\Th{\log_{2}{n}}$ span.
A similar approach works for FFT.

The technique described above works for all Strassen-like algorithms, including Victor Pan's $\Oh{n^{2.795}}$ work algorithm \cite{Pan78}. We remark that while we use additional techniques specific to the problems to achieve better work and span bounds, the main contribution is devising the general technique to enable limited work-sharing among the single-point computations using extra space.

Using extra space helps to avoid atomic TS operations, too. Blelloch et al. \cite{BlellochFiGuSu2019} achieve optimal span (with high probability) and optimal work (in expectation) for sorting, semisorting, and random permutation using TS. They use TS to distribute a set of $n$ items into $m = o(n)$ buckets where each item knows its destination bucket, but how many items will fall in a bucket is unknown. They reduce the problem to a variant of a balls and bins problem. In particular, when $n$ items try to find unoccupied cells randomly among $c\cdot n$ cells ($c$ is a constant) in parallel, it is enough for each item to try $\Th{\log n}$ times to find an unoccupied cell with high probability; the span of this process is thus $\Th{\log n}$. An item tries to put itself in a random cell using a TS on a flag to reserve it. If the TS fails, it tries again since the cell is already taken. In the absence of TS, it would take $\Th{\log n}$ synchronization time to figure out the items that fail to find a cell after each attempt, thus making the overall span $\Th{\log ^2 n}$ for the $\Th{\log n}$ attempts. Our approach can avoid this $\log n$ synchronization steps by increasing the space by a factor of $\log n$ and by allowing each item to simultaneously attempt to place itself in $\log n$ random cells. We assume arbitrary concurrent writes meaning that if multiple concurrent threads try to write to any given shared memory location simultaneously only one arbitrary thread succeeds. This approach does not increase span and work.

While Blelloch et al.'s~\cite{BlellochFiGuSu2019} sorting algorithm performs $\Th{n\log n}$ work in expectation, we achieve the same bound w.h.p. in $n$. The success probability of the distribution step in a recursive call of Blelloch et al.'s algorithm  is dependent on the size of the input to that recursive call. Since input size decreases doubly exponentially with the increase of recursion level, though the success probability is high in the corresponding input size, it reduces rapidly as execution moves deeper in the recursion tree and does not remain high w.r.t. the original input size $n_0$. 

To achieve a $\Th{n\log n}$ work bound w.h.p. in $n$ (along with a $\Th{\log{n}}$ span also w.h.p. in $n$) our {\MainSort} algorithm proceeds in two phases --- a recursive {\AlmostSort} phase and a non-recursive leftover integration phase. Extra space is used throughout {\MainSort} to simulate TS operations as described in the previous paragraph. Given an input of size $n_0$, {\AlmostSort} sorts $n_0 - o(n_0)$ items of the input, then the remaining $o(n_0)$ items are merged with the sorted $n_0 - o(n_0)$ items in the integration step. {\AlmostSort} is a recursive bucketing algorithm with some similarity to Blelloch et al.'s algorithm. However, the bucket size in each of its recursive calls is an $r(n_0)$ factor larger than the ones used in Blelloch et .al.'s paper, where $r(n_0) = \Th{\frac{\log n_0 \log \log \log n_0}{\log \log n_0}}$, and recursion in {\AlmostSort} is terminated much sooner than in Blelloch et .al.'s algorithm, so that throughout {\AlmostSort} partitioning succeeds w.h.p. in $n_0$. Each item tries to put itself into its destination bucket twice and fails with probability $1/(r(n_0))^2$. We set aside the failed items to be incorporated into the final sorted sequence during the integration step later and move to the next level of recursion without them. \hide{By leaving behind the failed items we make sure that the success probability of a distribution step is not tied to the size of the input given to the recursive call executing it.} By allowing some items to remain unsorted, we make the work done by each recursive step effectively independent of the probability that single items are successfully written to their chosen location. We show that with high probability in $n_0$, at most $\Th{n_0/(r(n_0))^2}$ items fail to move from any level of recursion to the next level. By switching to Cole-Ramachandran's deterministic sorting algorithm~\cite{ColeRa2017} after $\log\log\log n_0$ levels of recursion, we ensure that {\AlmostSort} performs $O(n_0 \log n_0)$ work w.h.p. in $n_0$. The integration phase then combines the $\Oh{n_0 \log\log\log{n_0} /(r(n_0))^2}$ (w.h.p. in $n_0$) leftover items with the already sorted sequence in $O(n_0 \log n_0)$ work w.h.p. in $n_0$.

\begin{table*}
\centering
\scalebox{0.9}{
{ 
\begin{colortabular}{ | l | l | l | l | l | l |}
\hline                       
%& & \multicolumn{3}{c|}{\cellcolor{tabletitlecolor} \idp{} } & \multicolumn{3}{c|}{\cellcolor{tabletitlecolor} \rdp{} } \\ \cline{3-6}
%\rowcolor{tabletitlecolor} & Work & Serial cache & Span & Parallelism \\  
%\rowcolor{tabletitlecolor} Problem & $(T_1)$ & comp. $(Q_1)$ & $(T_{\infty})$ & $(T_1 / T_{\infty})$ \\  \hline

Algorithm & Work $(T_1)$ & Space $(S_{\infty})$ & Span $(T_{\infty})$ & Result \\\hline
\rowcolor{lightred} \multicolumn{5}{|l|}{Strassen's Matrix Multiplication}\\ \hline

%AS & $\Th{n}$ & $\Th{1}$ & $\Th{n}$ & $\Th{ n/B + 1 }$ \\
%AS-HD & $\Th{n}$ & $\Th{r}$ & $\Th{n/r \times \log r}$ & $\Th{ n/B + n/r }$ \\
%AS-HS & $\Th{n}$ & $\Th{r}$ & $\Th{n/r + \log r}$ & $\Th{ n/B + r }$ \\
%AS-N & $\Th{n}$ & $\Th{n}$ & $\Th{\log n}$ & $\Th{ n/B + 1 }$ \\ \hline

Strassen's MM \cite{Strassen1969, CormenLeRiSt2009} & $\Oh{n^w}$ & $\Th{ n^w }$ & $\Oh{\log^2 n}$ &  \\

\faThumbsOUp{} \textsc{Strassen-S}  & $\Oh{n^w \log \log n}$ & $\Th{ n^w }$ & $\Oh{\log n}$ & Th. \ref{thm:strassen-s}
 \\
\faThumbsOUp{} \textsc{Strassen-W}  & $\Oh{n^w}$ & $\Th{ n^w /\log \log n }$ & $\Oh{\log n \log \log \log n}$ & Th. \ref{thm:strassen-w}\\
\faThumbsOUp{} \textsc{Strassen-S-Adaptive}  & $\Oh{n^w \log \log n}$ & $\Th{ s }$ & $\Oh{(n^w/s) \log n}$ & Th. \ref{thm:strassen-s-tunable}\\
\faThumbsOUp{} \textsc{Strassen-W-Adaptive}  & $\Oh{n^w}$ & $\Th{ s }$ & $\Oh{(n^w/s) \log^2 n}$ & Th. \ref{thm:strassen-w-tunable} \\

\hline
\rowcolor{lightred} \multicolumn{5}{|l|}{Sorting}\\ \hline
%\textsc{Sort-naive} & $\Oh{n \log n}$ & $\Th{ n \log n }$ & $\Oh{\log^2 n}$ &  \\
\textsc{Cole-Ramachandran} \cite{ColeRa2017} & $\Oh{n \log n}$ & $\Th{ n }$ & $\Oh{\log n \log \log n}$ &  \\
Blelloch et al. (atomic) \cite{BlellochFiGuSu2019} & $\Oh{n \log n}$ exp. & $\Th{ n }$ & $\Oh{\log n}$ whp &  \\
%\textsc{Sort-$n^{\epsilon}$way} & $\Oh{(1/\epsilon) n^{1 + \epsilon}}$ & $\Th{ (1/\epsilon) n }$ & $\Oh{(1/\epsilon + \epsilon) \log n}$ & \faThumbsOUp{} \\
 \faThumbsOUp{} {\MainSort}  & $\Oh{n \log n}$ whp & $\Th{ \frac{n \log n \log \log \log n}{\log \log n}}$ & $\Oh{\log n} $ whp & Th. \ref{thm:sorting-main-sort} \\
\faThumbsOUp{} {\TunableSort}  & $\Oh{n \log n} $ whp & $\Th{ s }$ & $\Oh{(n/s)\log^2 n} $ whp & Th. \ref{thm:sorting-tunable-sort} \\

\hline
\rowcolor{lightred} \multicolumn{5}{|l|}{Fast Fourier Transform}\\ \hline
%\textsc{FFT-2way} & $\Oh{n \log n}$ & $\Th{ n }$ & $\Oh{\log^2 n}$ &  \\
Cooley-Tukey $\sqrt{n}$-way \cite{CooleyTu65} & $\Oh{n \log n}$ & $\Th{ n }$ & $\Oh{\log n \log \log n}$ &  \\
%\textsc{FFT-opt} & $\Oh{(1/\phi) n^{1 + \phi}}$ & $\Th{(1/\phi) n  }$ & $\Oh{(1/\phi) \log n}$ & \faThumbsOUp{} \\
\faThumbsOUp{} \textsc{FFT} & $\Oh{n \log^{g(n)} n}$ & $\Th{ n }$ & $\Oh{\log n \log \log \log n}$ & Cor. \ref{cor:fft}\\

\hline

\end{colortabular}
}
}
\caption{{\small The complexity analyses of our new algorithms (marked with \faThumbsOUp{}) in the binary-forking model without atomics. Here, $n=$ problem size, $s=$ space $\in [\text{input size}, \text{work}]$, $w = \log_2 7$, exp. $ = $ expected time, whp $ = $ with high probability in $n$, $g(n) < 2$ for $n < 10^{10,000}$. \hide{The running time is $T_p = \Oh{T_1 / p + T_{\infty}}$ when we use a randomized work-stealing scheduler on a $p$-processor parallel machine with private caches.}\hide{$M=$ cache size, $B=$ block size, $f = (rn^2/M)^{1/3}$,  , $r=$\#copies of the output matrix/tensor, For tensor contraction, $X_{u \times v} = U_{u \times x} V_{x \times v}$, $z = u + v + x$.}
}}
\vspace{-0.3cm}
\label{tab:summary}
\end{table*}  

\paragraph{Binary-Forking Model.} Binary-forking model captures the current multi-core shared-memory systems. Many parallel algorithms are based on binary-forking model~\cite{acar2000data, ben2016parallel, blelloch2008provably, blelloch2011scheduling, das2019data}. Computations in the binary-forking model can be viewed as a series-parallel DAG where each node represents a thread's instruction. The root of the tree is the first instruction of the starting thread. Each node has at most two children. If node $u$ denotes the $i$-th instruction of thread $t$ and $u$ has only one child $v$, then $v$ denotes the $(i+1)$-th instruction of thread $t$. If node $u$ has two children $v$ and $w$, then $v$ represents the $(i+1)$-th instruction of thread $t$ and $w$ represents the first instruction of the new forked thread $t^{'}$. The binary-forking model includes ``join'' instructions to join the forking threads. They are modeled as a node with two incoming edges. The work of the computation is the number of nodes in the series-parallel DAG and the span of the computation is the length of the longest path in the DAG assuming unbounded resources such as processors and space.

\paragraph{Performance Metrics of a Parallel Program.} We use the \textit{work-span} model \cite{CormenLeRiSt2009} to analyze the performance of parallel programs executed on shared-memory multicore machines. The \textit{work} of a multithreaded program, denoted by $T_1(n)$, where $n$ is the input parameter, is defined as the total number of CPU operations it performs when executed on a single processor. The \textit{span} $T_{\infty}(n)$ of a program which is also known as its \textit{critical-path length} or \textit{depth}, is the maximum number of operations performed on any single processor when the program is run on an unbounded number of processors. The \textit{parallel running time} $T_p(n)$ of a program when run on $p$ processors under a greedy scheduler is given by $T_p(n) = \Oh{T_1(n) / p + T_{\infty}(n)}$. The \textit{parallelism}, computed by the ratio of $T_1(n)$ and $T_{\infty}(n)$, is the average amount of work performed by the program in each step of its critical path.

\hide{
Several parallel algorithms have been designed for the binary-forking model without locks and atomic instructions \cite{frigo2012cache, CormenLeRiSt2009, ChowdhuryRa2006, ChowdhuryRa2008, blelloch2008provably, ChowdhuryRa2010, BlellochGiSi2010, blelloch2011scheduling, TithiMaMeCh2013, chowdhury2016autogen, ChowdhuryGaTsTiBaLeSoKuTa2017, ChowdhuryGaTaTi2017, JavanmardGaDaAhTsCh2019, das2019data,  TithiGaTaAgCh2015, Tithi2015, Ganapathi2016, TangYoKaTiGaCh2014}. 
}

\section{Strassen's Matrix Multiplication}
\label{sec:strassen}

Suppose $w = \log_2 7$. Strassen's matrix multiplication (MM) algorithm \cite{Strassen1969} performs $\Oh{n^{w}}$ work (i.e., multiplications and additions), unlike the classic MM algorithm that performs $\Oh{n^3}$ work. A straightforward parallelization of Strassen's MM leads to $\Th{\log^2 n}$ span. Our goal is to design a parallel Strassen's MM in the binary-forking model without using locks and atomic instructions to achieve an optimal span of $\Oh{\log n}$ without affecting the work bound of $\Th{n^{w}}$. 

In this paper, we present parallel Strassen MM algorithms $(i)$ having optimal $\Oh{\log n}$ span and $\Oh{n^{w} \log \log n}$ work, i.e., work very close to that of the standard Strassen's MM; and $(ii)$ having $\Oh{n^{w}}$ work and $\Oh{\log n \log \log \log n}$ span, i.e., very close to optimal span.

The core ideas and techniques used in our algorithms are as follows. We first perform \highlight{single-point computation}, i.e., computation of a single cell of the output matrix independently from that of other cells/entries. This implies that all cells of the output matrix are computed independently in $\Oh{\log n}$ span. However, as there is no work-sharing across multiple threads, the total work blows up to $\Oh{n^{w + 2}}$. We enable partial work-sharing across threads by saving intermediate computations at carefully selected levels of recursion. By splitting the recursion tree into \highlight{stages} and allowing work-sharing across stages, we are able to reduce the work to very close to $\Oh{n^w}$. Hence, by using single-point computations in stages, we are able to obtain good work and span bounds. We use this algorithm to design other parallel Strassen's MM algorithms with different advantages.

\hide{We first describe the $k$-way Strassen algorithm, a simple generic parallel Strassen algorithm with branching factor at most $k$. We then present our parallel Strassen algorithm performing work that matches that of the serial Strassen algorithm and achieves optimal span. We use this algorithm to design other parallel Strassen's MM algorithms with different advantages.}

\paragraph{$k$-way Strassen's MM \cite{Strassen1969,CormenLeRiSt2009}.} The $k$-way Strassen's MM, for $k \in [1, 7]$, executes the child nodes in exactly $\ceil{7/k}$ parallel steps without executing more than $k$ child nodes at a time.

\begin{lemma}[\cite{Strassen1969,CormenLeRiSt2009}]
\label{thm:7way}
The $k$-way Strassen's MM has a complexity of $\Oh{n^w}$ work, $\Oh{\log^2 n}$ span if $k = 7$, $\Oh{n^{\log_2 \ceil{7/k}}}$ span if $k \ne 7$, $\Oh{n^2 \log n}$ space if $k = 4$, and $\Oh{n^{\max{(2, \log_2 k)}}}$ space if $k \ne 4$.
\hide{\textsc{Strassen-1way} does $\Oh{n^w}$ work, has $\Oh{n^w}$ span, and uses $\Oh{n^2}$ space. \textsc{Strassen-7way} does $\Oh{n^w}$ work, has $\Oh{\log^2 n}$ span, and uses $\Oh{n^w}$ space.} 
\end{lemma}
\begin{proof}
The work, span, and extra space recurrences for the $k$-way Strassen's MM are as follows. If $n = 1$, then $T_{1}(n) = \Oh{1}$ and $T_{\infty}(n) = \Oh{1}$. If $n > 1$, then
\begin{align*}
&T_{1}(n) = 7 T_{1}(n/2) + \Oh{n^2}, &&T_{\infty}(n) =  \ceil{7/k} T_{\infty}(n/2) + \Oh{\log n}, &&&S_{\infty}(n) = k S_{\infty}(n/2) + \Oh{n^2}.
\end{align*}

Solving these recurrences, we have the lemma.
\hide{\textcolor{red}{from meeting:
$$S_\infty = \begin{cases}
\Oh{n^2} & \text{if } k < 4 \\
\Oh{n^2 \log n} & \text{if } k = 4 \\
\Oh{n^{\log_2 k}} & \text{if } k > 4
\end{cases},$$
$$T_\infty = \begin{cases}
\Oh{\log^2 n} & \text{if } k = 7 \\
\Oh{n^{\log_2 \ceil{7/k}}} & \text{else}
\end{cases}$$
$$T_1(n) = S_\infty(n) \biggr|_{k=7} = \Oh{n^{\log_2 k}}$$
}
}
\end{proof}

The work of the $k$-way Strassen's MM for any value of $k$ is $\Oh{n^{\log_2 7}}$. The $k$-way algorithm gives a trade-off between span and space. When $k = 1$, we get the standard Strassen's algorithm \cite{Strassen1969}. When $k = 7$, we get the standard parallel Strassen's MM \cite{CormenLeRiSt2009} that spawns all the child nodes in parallel achieving $\mathcal{O}(\log^2 n)$ span and occupying $\Oh{n^{\log_2 7}}$ space. 

\begin{figure*}[!ht]
\centering
\scalebox{0.94}{
\begin{minipage}{\textwidth}
\begin{mycolorbox}{$\textsc{Strassen-S}(X, U, V)$ \xcomment $X \gets U \times V$}
\begin{minipage}{0.99\textwidth}
{\codesize
\algotopspace{}
\noindent
\begin{enumerate}
\setlength{\itemindent}{-1.5em}

\vsitem ($\overline{U}, \overline{V}, \overline{X}$, $U$quad, $V$quad, $X$branch) $\gets$ \textsc{Construct-Data-Structures}($U, V$)
\vsitem \textsc{Compute-Input-Matrices}($\overline{U}, 0, 0$, $U$quad);  \textsc{Compute-Input-Matrices}($\overline{V}, 0, 0$, $V$quad)
\vsitem \textsc{Compute-Output-Matrices}($\overline{X}, 0, 0$, $X$branch)
\vsitem $X \gets \overline{X}[0][0]$ 

\algobottomspace{}
\end{enumerate}
}
\end{minipage}
\end{mycolorbox}
\vgap{}
\begin{mycolorbox}{$\textsc{Compute-Input-Matrices}(Z, \text{stage\_id, root\_id, quad})$}
\begin{minipage}{0.99\textwidth}
{\codesize
\algotopspace{}
\noindent
\begin{enumerate}
\setlength{\itemindent}{-1.5em}

\vsitem height $\gets$ \#levels in the stage; \#leaves $\gets 7^{\text{height}}$ 
\vsitem \xparallelfor node $\gets 0$ \xto \#leaves $-1$ \xdo \label{line:leafnodeslaunching}
\vsitem \T leaf\_id $\gets$ (root\_id $- 1$) $\times$ \#leaves $+ $ node 
\vsitem \T \xparallelfor $i \gets 0$ \xto $n-1$ \xdo \label{line:spawninputrows}
\vsitem \T \T \xparallelfor $j \gets 0$ \xto $n-1$ \xdo \label{line:spawninputcols}
\vsitem \T \T \T $Z$[stage\_id][leaf\_id][$i, j$] $\gets$  \textsc{Compute-Input-Cell}($Z$, stage\_id, leaf\_id, $i, j, n$, height, quad)
\vsitem \T \xif not last stage \xthen \textsc{Compute-Input-Matrices}($Z$, stage\_id + 1, leaf\_id, quad)
		
\algobottomspace{}
\end{enumerate}
}
\end{minipage}
\end{mycolorbox}
\vgap{}
\begin{mycolorbox}{$\textsc{Compute-Input-Cell}(Z, \text{stage\_id, node\_id, $i, j, n$, height, quad})$}
\begin{minipage}{0.99\textwidth}
{\codesize
\algotopspace{}
\noindent
\begin{enumerate}
\setlength{\itemindent}{-1.5em}

\vsitem \xif height $= 0$ \xthen \xreturn $Z$[stage\_id][node\_id][$i, j$]
\vsitem parent\_id $\gets$ (node\_id $/$ 7), branch\_id $\gets$ node\_id mod 7
\vsitem \xparallelfor $k \gets 0$ \xto $1$ \xdo
\vsitem \T $t[k] \gets 0$; coeff $\gets$ quad[branch\_id][$k$].coeff
\vsitem \T \xif coeff $\ne 0$ \xthen
\vsitem \T \T new\_$i \gets n \times $ quad[branch\_id][$k$].shift\_$i + i$; new\_$j \gets n \times $ quad[branch\_id][$k$].shift\_$j + j$
\vsitem \T \T $t[k] \gets$ coeff $\times$ \textsc{Compute-Input-Cell}($Z$, stage\_id, parent\_id, new\_$i$, new\_$j$, $2n$, height $-1$, quad)

\vsitem \xreturn $t[0] + t[1]$

\algobottomspace{}
\end{enumerate}
}
\end{minipage}
\end{mycolorbox}
\vgap{}
\begin{mycolorbox}{$\textsc{Compute-Output-Matrices}(Z, \text{stage\_id, root\_id, branch})$}
\begin{minipage}{0.99\textwidth}
{\codesize
\algotopspace{}
\noindent
\begin{enumerate}
\setlength{\itemindent}{-1.5em}

\vsitem height $\gets$ \#levels in the stage; \#leaves $\gets 7^{\text{height}}$ 

\vsitem \xparallelfor node $\gets 0$ \xto \#leaves $-1$ \xdo
\vsitem \T leaf\_id $\gets$ (root\_id $- 1$) $\times$ \#leaves $+ $ node 
\vsitem \T \xif last stage \xthen $Z$[stage\_id][leaf\_id][0,0] $\gets$ $U$[stage\_id][leaf\_id][0,0] $\times$ $V$[stage\_id][leaf\_id][0,0]
\vsitem \T \xelse \textsc{Compute-Output-Matrices}$(Z, \text{stage\_id + 1, leaf\_id, branch})$

\vsitem \xparallelfor $i \gets 0$ \xto $n-1$ \xdo \label{line:spawnoutputrows}
\vsitem \T \xparallelfor $j \gets 0$ \xto $n-1$ \xdo \label{line:spawnoutputcols}
\vsitem \T \T  \textsc{Compute-Output-Cell}($Z$, stage\_id, root\_id, $i, j, n$, height, branch)

\algobottomspace{}
\end{enumerate}
}
\end{minipage}
\end{mycolorbox}
\vgap{}
\begin{mycolorbox}{$\textsc{Compute-Output-Cell}(Z, \text{stage\_id, node\_id, $i, j, n$, height, branch})$}
\begin{minipage}{0.99\textwidth}
{\codesize
\algotopspace{}
\noindent
\begin{enumerate}
\setlength{\itemindent}{-1.5em}

\vsitem \xif height $= 0$ \xthen \xreturn $Z$[stage\_id + 1][node\_id][$i, j$]
\vsitem shift\_$i \gets [i > n/2]$; shift\_$j \gets [j > n/2]$ \xcomment{[ ] is the Iversion bracket}
\vsitem quad\_id $\gets 2 $ shift\_$i + $ shift\_$j$; new\_$i \gets i - (n/2) $ shift\_$i$; new\_$j \gets j - (n/2) $ shift\_$j$

\vsitem \xparallelfor $k \gets 0$ \xto $6$ \xdo
\vsitem \T $t[k] \gets 0$; coeff $\gets$ branch[quad\_id][$k$]
\vsitem \T \xif coeff $\ne 0$ \xthen
\vsitem \T \T child\_id $\gets$ (node\_id $- 1$) $\times 7 + k$
\vsitem \T \T $t[k] \gets$ coeff $\times$ \textsc{Compute-Output-Cell}($Z$, stage\_id, child\_id, new\_$i$, new\_$j$, $n/2$, height $- 1$, branch)
\vsitem \xreturn $\sum_{k = 0}^{6} t[k]$

\algobottomspace{}
\end{enumerate}
}
\end{minipage}
\end{mycolorbox}
\end{minipage}
}
\vgap{}
\caption{The \textsc{Strassen-S} MM algorithm.}
\label{fig:strassen-algorithm}
\end{figure*}

\begin{figure*}[!ht]
\centering
\begin{minipage}{\textwidth}
{\scriptsize
\centering
\begin{minipage}{0.6\textwidth}
% Uquad
\centering
\begin{colortabular}{c | c c}
 & $k = 0$ & $k = 1$\\
\rowcolor{tabletitlecolor} \rotatebox{90}{branch\_id} & $\langle$ \rotatebox{90}{shift\_i}, \rotatebox{90}{shift\_j}, \rotatebox{90}{coeff} $\rangle$ & $\langle$ \rotatebox{90}{shift\_i}, \rotatebox{90}{shift\_j}, \rotatebox{90}{coeff} $\rangle$ \\ \hline
$0$ & $\langle 0, 0, 1 \rangle$ & $\langle 1, 1, 1 \rangle$ \\
$1$ & $\langle 1,0,1 \rangle$ & $\langle 1,1,1 \rangle$ \\
$2$ & $\langle 0,0,1 \rangle$ & $\langle -,-,0 \rangle$ \\
$3$ & $\langle 1,1,1 \rangle$ & $\langle -,-,0 \rangle$ \\
$4$ & $\langle 0,0,1 \rangle$ & $\langle 0,1,1 \rangle$ \\
$5$ & $\langle 1,0,1 \rangle$ & $\langle 0,0,-1 \rangle$ \\
$6$ & $\langle 0,1,1 \rangle$ & $\langle 1,1,-1 \rangle$ \\
\end{colortabular}
%Vquad
\begin{colortabular}{c|cc}
 & $k = 0$ & $k = 1$\\
\rowcolor{tabletitlecolor} \rotatebox{90}{branch\_id} & $\langle$ \rotatebox{90}{shift\_i}, \rotatebox{90}{shift\_j}, \rotatebox{90}{coeff} $\rangle$ & $\langle$ \rotatebox{90}{shift\_i}, \rotatebox{90}{shift\_j}, \rotatebox{90}{coeff} $\rangle$ \\ \hline
$0$ & $\langle 0,0,1 \rangle$ & $\langle 1,1,1 \rangle$ \\
$1$ & $\langle 0,0,1 \rangle$ & $\langle -,-,0 \rangle$ \\
$2$ & $\langle 0,1,1 \rangle$ & $\langle 1,1,-1 \rangle$ \\
$3$ & $\langle 1,0,1 \rangle$ & $\langle 0,0,-1 \rangle$ \\
$4$ & $\langle 1,1,1 \rangle$ & $\langle -,-,0 \rangle$ \\
$5$ & $\langle 0,0,1 \rangle$ & $\langle 0,1,1 \rangle$ \\
$6$ & $\langle 1,0,1 \rangle$ & $\langle 1,1,1 \rangle$ \\
\end{colortabular}
\end{minipage}
\begin{minipage}{0.4\textwidth}
% Xbranch
\centering
\begin{colortabular}{c|ccccccc|}
& \multicolumn{7}{c|}{$k$} \\
\rowcolor{tabletitlecolor}  \rotatebox{90}{quad\_id} & 0 & 1 & 2 & 3 & 4 & 5 & 6 \\ \hline
$0$ & $1$ & $0$ & $0$ & $1$ & $-1$ & $0$ & $1$\\
$1$ & $0$ & $0$ & $1$ & $0$ & $1$ & $0$ & $0$\\
$2$ & $0$ & $1$ & $0$ & $1$ & $0$ & $0$ & $0$\\
$3$ & $1$ & $-1$ & $1$ & $0$ & $0$ & $1$ & $0$\\
\end{colortabular}
\end{minipage}
}
\end{minipage}
\vgap{}
\caption{Data structures required for the \textsc{Strassen-S} MM algorithm. Left: $Uquad$ and $Vquad$. Right: $Xbranch$.}
\label{fig:strassen-algorithm-ds}
\end{figure*}

\paragraph{\textsc{Strassen-S} MM.} In this section, we present a parallel Strassen's MM algorithm, as shown in Figure \ref{fig:strassen-algorithm}, that achieves the optimal span of $\Oh{\log n}$ with only a $\Oh{\log \log n}$ factor increase in the work compared with the classical sequential Strassen's MM algorithm. In this algorithm, we multiply two matrices $U$ and $V$ and store the matrix product in $X$. We first construct the required data structures as shown in Figure \ref{fig:strassen-algorithm-ds}. We then compute the input matrices ($U$ and $V$) at all nodes in the recursion tree in parallel in $\Oh{\log n}$ span. Finally, we compute the output matrix ($X$) at all nodes in the recursion tree in $\Oh{\log n}$ span. 

\hide{
Before we present our algorithm, let's look at the data layout of our algorithm. Figure \ref{fig:strassen-algorithm} shows our algorithm.

In this step, we allocate space for auxiliary matrices for all nodes in the recursion tree. Suppose that the levels in the recursion tree are $0, 1, \ldots, (\ell = \log n - 1)$. We first create an $\ell$-sized array such that its $i$th entry points to a $7^i$-sized index array, as shown in Figure \ref{fig:indexing}. The indices in this index array is stored in the base-7 system. The $j$th index in the $i$th index array represents the $j$th node at the $i$th level in the entire recursion tree as per the BFS ordering. The index further points to the auxiliary matrices ($U$, $V$, and $X$) for the input and output at that particular node. In this way, every node in the recursion tree can be indexed in parallel in $\Oh{\log n}$ span.

\begin{figure}[!ht]
\centering
\includegraphics[width=0.3\textwidth]{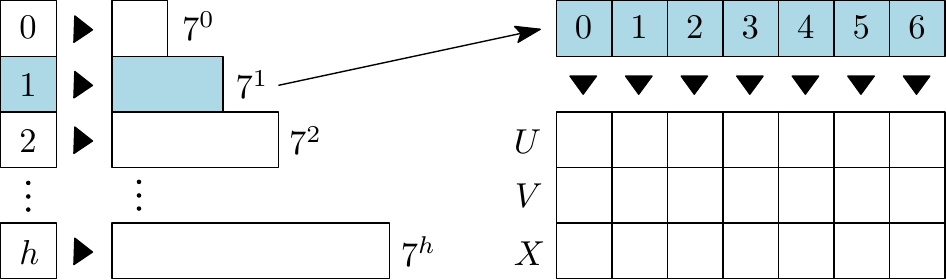}
\vgap{}
\caption{\small Data layout for the indices and the auxiliary matrices for all nodes in the recursion tree. Here, $\ell = \log n$.}
\label{fig:indexing}
\vgap{}
\end{figure}
}

\vspace{0.2cm}
\noindent
\textbf{[Step 1. Compute the Input Matrices.]} Consider the standard 7-way parallel Strassen's MM. The height of the recursion tree is $\Oh{\log n}$ and in each level, the total cost of forking and synchronizing threads to compute the input matrices is $\Oh{\log n}$. Hence, the total span for computing input matrices at all nodes in the recursion tree is $\Oh{\log^2 n}$. We can reduce the span to $\Oh{\log n}$ using single-point computation.

A cell of an input matrix ($U$ or $V$) at a node of the recursion depends on at most two cells of the corresponding input matrix at its parent node. This implies that each cell in an input matrix at a leaf node depends on at most $2^{\log n} = n$ cells in the corresponding input matrix at the root node. If we were to compute all input cells of all input matrices at all nodes, the total work would explode to $\Oh{n^w \times n} = \Oh{n^{w+1}}$. To keep the work very close to $\Oh{n^w}$, we split the entire recursion tree into stages. We then use single-point computation of input cells in stages.

For this algorithm, we have $\Oh{\log \log n}$ stages so that the work performed in each stage is $\Oh{n^w}$. Using single-point computation in each stage, we are able to achieve the desired optimal span of $\Oh{\log n}$ limiting the total work to $\Oh{n^w \log \log n}$.

In this step, we compute the input matrices of all nodes in the recursion tree. The step consists of $h+1$ sequential stages: $0, 1, \ldots, h$, as shown in Figure \ref{fig:strassen-stages}, such that the height of stage $i$ is fixed at $c_i \log n$, where $h$ and $c_i$ are given below:
\begin{align}
&c_i =  \begin{cases}
0 & \text{if } i = - 1,\\
1-\alpha^{i + 1} & \text{if } i \in [0, h - 1],\\
1 & \text{if } i = h.
  \end{cases} \text{ such that } w = \log_2 7, \alpha = \frac{1}{w-1}, \text{ and } h = \log_{w - 1} \frac{\log n}{\log \log \log n}. \label{eq:stages-bounds}
\end{align}

\begin{figure}[!ht]
\centering
\includegraphics[width=0.7\textwidth]{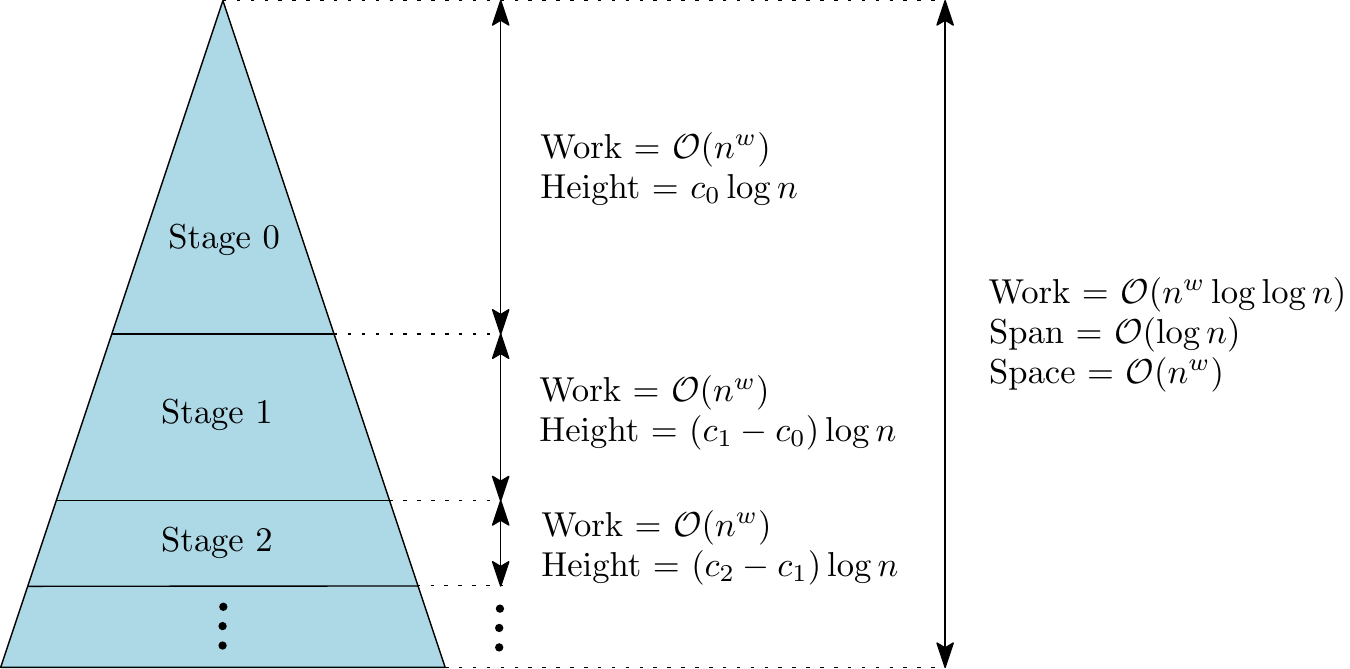}
\caption{Stages in the \textsc{Strassen-S} MM algorithm for computing the input matrices.}
\label{fig:strassen-stages}
\end{figure}

Please refer to Figure \ref{fig:strassen-algorithm} (the \textsc{Compute-Input-Matrices} algorithm) for computing the input matrices ($U$ and $V$) for all the leaf nodes in all stages. We start from stage 0. For any given stage, denoted by stage id, we can easily compute the topmost level, called the root level and the bottommost level, called the leaf level. It is also easy to list out all indices of the leaf nodes in a given stage. So, for all leaf nodes, for all cells in the input matrix in a particular leaf node, we invoke the function \textsc{Compute-Input-Cell}. This function computes the value of a specific cell in the input matrix of a leaf node.

\begin{figure}[!ht]
\centering
\begin{minipage}{\textwidth}
\begin{minipage}{0.5\textwidth}
\includegraphics[width=\textwidth]{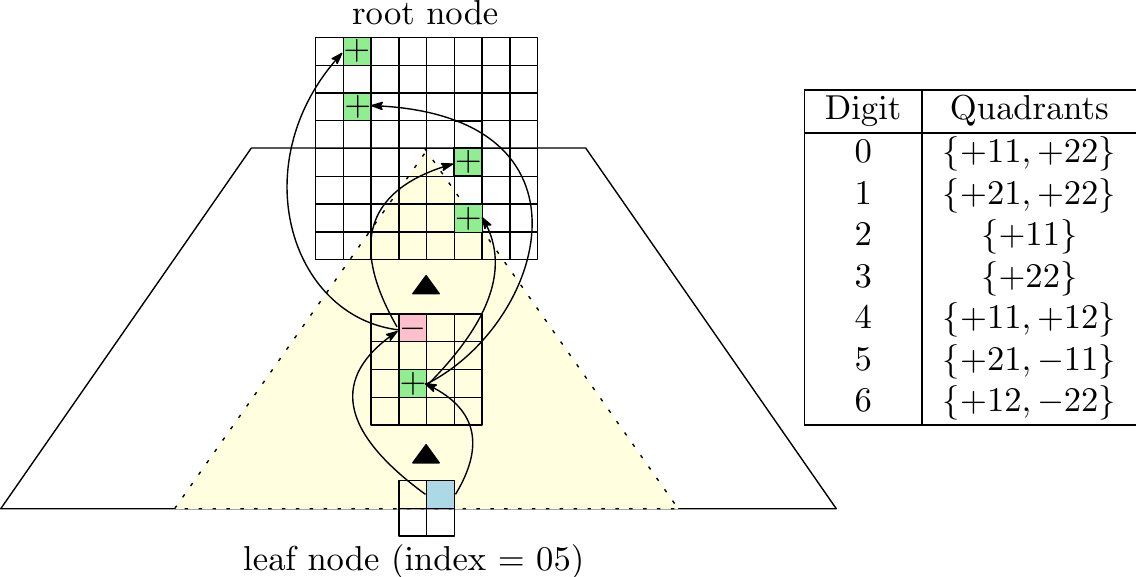}
\end{minipage}
\begin{minipage}{0.5\textwidth}
\includegraphics[width=\textwidth]{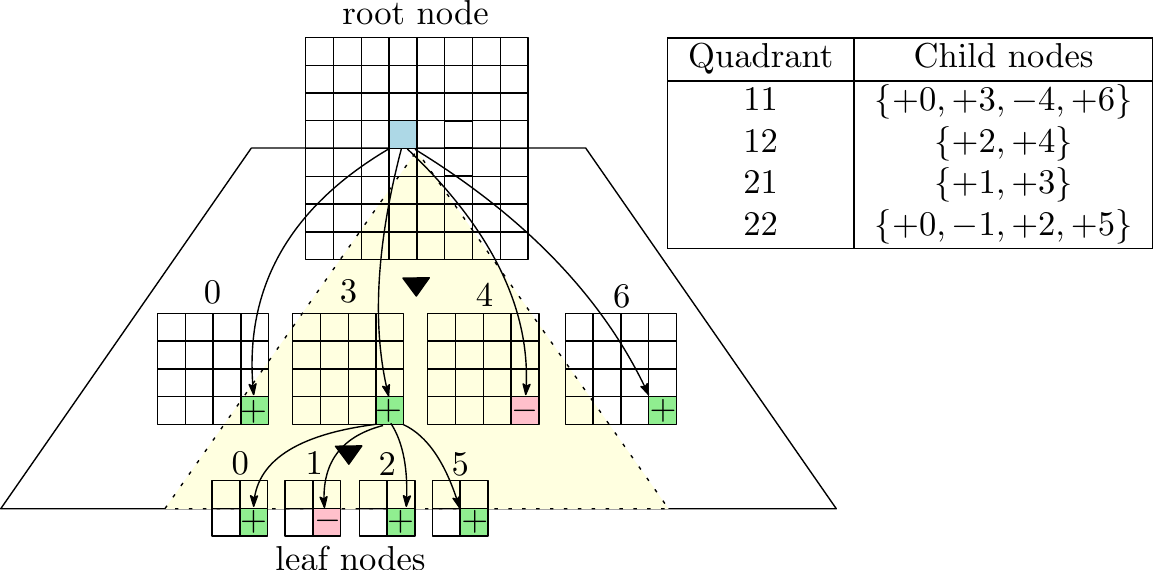}
\end{minipage}
\end{minipage}
\caption{\small Left: Single-point computation of a cell in the input matrix $U$ at a leaf node in a stage. Right: Single-point computation of a cell in the output matrix $X$ at a root node in a stage. (If there is an arrow from cell $a$ to cell $b$, it means that cell $a$ depends on cell $b$.)}
\label{fig:strassen-singlepoint-computation}
\end{figure}

The working of the \textsc{Compute-Input-Cell} algorithm is as shown in Figure \ref{fig:strassen-singlepoint-computation} (left). The figure shows the way in which a highlighted cell in the $U$ matrix at a leaf node with id 05 (in base-7 system, for simplicity) is computed. As the last digit of the index is 5, it means that the leaf node is the 5th child of its parent. From the logic of the Strassen's MM algorithm, we know that the $U$ matrix in the 5th child node of a parent node is computed by subtracting the first quadrant $(-11)$ from the third quadrant $(+21)$ of the $U$ matrix of the parent node. Hence, we can exactly know the two cells in the $U$ matrix of the parent node on which the highlighted cell in the $U$ matrix of the leaf node depends. Also, we can compute the highlighted cell in $\Oh{1}$ time using the two cells of the parent node. The first digit of the index of the lead node is 0. This means that the parent node of the leaf node is the 0th child of its parent (i.e., the leaf node's grandparent). From the logic of the Strassen's MM algorithm, we know that the $U$ matrix in the 0th child node of a parent node is computed by adding the first quadrant $(+11)$ to the fourth quadrant $(+22)$ of the $U$ matrix of the parent node. Using this approach, we can trace the path from the leaf node to its ancestor at the root level. So, each cell in the leaf node depends on 2 cells in its parent node which in turn depends on 4 cells in its parent node and so on until we reach a node at the root level. In this way, we can spawn multiple threads that recursively compute each cell at the leaf node using cells from its ancestor at the root level of the stage. The span for computing each cell is simply the height of the stage i.e, the number of levels in that stage.

Once all the cells in a leaf node with id 05 are computed, the algorithm recursively and asynchronously invokes \textsc{Compute-Input-Matrices} for the next stage with this leaf node as the new root. The base case of the \textsc{Compute-Input-Matrices} algorithm is when the algorithm reaches the last stage at which we compute the cells at the leaf nodes using the exact same idea.  

\begin{lemma}
\label{lem:computing-input-matrices}
\textsc{Compute-Input-Matrices} has a complexity of $\Oh{n^w \log \log n}$ work, $\Oh{\log n}$ span, and $\Oh{n^w}$ space.
\end{lemma}
\begin{proof} $[$Work.$]$ We compute the input matrices in $h+1$ \bands $S_0, S_1, \ldots, S_{h}$, where\\
	$h = \log_{w -1}(\log n/\log \log \log n)$. Suppose $W_i$ define the work done at \band $S_i$. We first come up with a generic formula for $W_i$. We use a direct proof to show that $W_i = \Oh{n^w}$, which implies that the total work is $\sum_{i = 0}^{h} W_i = \Oh{n^w \log \log n}$.
	
	We compute $W_i$ for $i \in [0, h - 1]$. \#Nodes at the leaf level of \band $S_i$ is $7^{c_{i}\log n} = n^{wc_i}$. The \#cells in a matrix at the leaf level is $(n/2^{c_i \log n})^2 = n^{2(1-c_i)}$. Each cell in a matrix at the leaf level depends on $\mathcal{O}(2^{(c_{i}-c_{i-1})\log n}) = \Oh{n^{c_{i}-c_{i-1}}}$ cells in a matrix at the root level of the stage. Hence, $W_i = \Oh{n^{w c_i} n^{2(1-c_i)} n^{c_{i}-c_{i-1}}}$ $=$ $\Oh{n^{c_i (w - 1) - c_{i-1} + 2}}$. To show that $W_i = \Oh{n^w}$ for all $i \in [0, h-1]$, it is enough to prove that $c_i (w - 1) - c_{i-1} + 2 = w$. We substitute the values of $c_i$ and $c_{i - 1}$ from equation \ref{eq:stages-bounds} to get: $c_i (w - 1) - c_{i-1} + 2 = (1 - (1/(w-1))^{i+1}) (w - 1) - (1 - (1/(w-1))^i) + 2 = w$.
	
	We now compute $W_h$. The height of the first $h$ stages is $c_{h-1} \log n$. So, the height of the last stage $S_h$ is $\log n - c_{h - 1} \log n$. Substituting the value of $c_{h-1}$ from equation \ref{eq:stages-bounds} and simplifying, we get the height of stage $S_h$ as $\log \log \log n$. There are $n^w$ nodes at $S_h$. The size of a matrix at a leaf node is $1 \times 1$. Each cell depends on $2^{ \log\log\log n} = \log\log n$ cells in a matrix at the root level of stage $S_h$. Hence, work done at the last stage is $W_h = \Oh{n^w \log\log n}$.
	
	Combining the work of the first $h$ \bands and the last \band, we get $T_1(n) = \sum_{i = 0}^{h-1} W_i + W_h = \Oh{n^w \log\log n}$.
	
	\vspace{0.1cm}
	\noindent
	[Span.] Let $T_{\infty}(m, i)$ denote the span of the \textsc{Compute-Input-Matrices} algorithm starting from stage $i$ where a matrix at the root level is of size $m \times m$. We give a recursive formula to compute $T_{\infty}(m, i)$. Then, the total span for the algorithm is $T_{\infty}(n, 0)$.
	
	Consider the \textsc{Compute-Input-Matrices} algorithm.  Let $\Delta c_i = c_i - c_{i - 1}$.
	\#Nodes at the leaf level of stage $S_i$ is $7^{\Delta c_i \log n}$. Launching these nodes in parallel (line \ref{line:leafnodeslaunching}) incurs a span of $\Oh{\Delta c_i w\log n}$. A matrix at the leaf level will be of size $m/(2^{\Delta c_i\log n}) \times m/(2^{\Delta c_i\log n})$. Spawning \textsc{Compute-Input-Cell} function for all cells (lines \ref{line:spawninputrows}, \ref{line:spawninputcols}) incur a span of $\Oh{2\log m - 2\Delta c_i\log n}$. Executing the \textsc{Compute-Input-Cell} algorithm incurs $\Oh{\Delta c_i\log n}$ span. Adding all these spans give us $\Oh{\Delta c_i \log n + \log m}$.
	
	The span of stage $S_i$ recursively depends upon the span of stage $S_{i + 1}$ as the matrices at the leaf level of stage $S_{i + 1}$ are constructed from the leaf level matrices of stage $S_{i}$. Hence, $T_{\infty}(m,i)$ can be recursively defined using the previous analysis as: $T_{\infty}(m,i) = \Oh{\log \log \log n}$ if $i = h$ and $T_{\infty}(m,i) = \Oh{\Delta c_i \log n + \log m} + T_{\infty}(m^{1 - \Delta c_i},i+1)$ if $i < h$. Substituting the values of $c_i$ from equation \ref{eq:stages-bounds}, we get $\Delta c_i = 1 - \alpha^{i+1} - (1 - \alpha^{i}) = \alpha^{i}(1-\alpha) = \Oh{\alpha^i}$. We know that $m$ starts with $n$ and decreases by a factor of $n^{\Delta c_i}$ for every stage. Hence, $m = n^{1-c_{i-1}} = n^{\alpha^i}$, which implies that $\log m = \alpha^i \log n$.
	
	By unrolling the recursion and using the fact that $\alpha^{i}$ is a geometric series and $\alpha < 1$, we compute the total span as $T_{\infty}(n,0) = \sum_{i = 0}^{h-1} \alpha^i \log n + T_{\infty}(n,h) = \Oh{\log n}$.
	
	\vspace{0.1cm}
	\noindent
	[Space.] The total space is dominated by the space used by the data structures. There are $n^w$ matrices at the leaf level for each of the input matrices $U$ and $V$. Each such matrix is of size $1 \times 1$.  Hence, space usage is $\Oh{n^w}$.
\end{proof}

\noindent
\textbf{[Step 2. Compute the Output Matrices.]} 
The logic used to compute the output matrices is very similar to that used to compute the input matrices. A cell of the output matrix ($X$) at a node of the recursion depends on at most four cells of the corresponding output matrices at its child nodes. This implies that each cell in the output matrix at the root node depends on at most $4^{\log n} = n^2$ cells in the corresponding output matrices at the leaf nodes. If we were to compute all output cells of all output matrices at all nodes, the total work would explode to $\Oh{n^w \times n^2} = \Oh{n^{w+2}}$. To keep the work very close to $\Oh{n^w}$, we split the entire recursion tree into stages. We then use single-point computation of output cells in stages.

In this step, we compute the output matrix of all nodes in the recursion tree. The phase consists of $h+1$ sequential stages: $0, 1, \ldots, h$, as shown in Figure \ref{fig:strassen-stages} (replace $c_i$'s with $d_i$'s), such that the height of stage $i$ is fixed at $d_i \log n$, where $h$ and $d_i$ are given below:
\begin{align}
&d_i =  \begin{cases}
0 & \text{if } i = - 1,\\
1-\beta^{i + 1} & \text{if } i \in [0, h - 1],\\
1 & \text{if } i = h.
  \end{cases} \text{ such that } w = \log_2 7, \beta =\frac{4-w}{2}, \text{ and } h = \log_{ \frac{2}{4-w}} \frac{2\log n}{\log \log \log n} -1. \label{eq:stages-bounds-2}
\end{align}

In step 1, we computed the input matrices in the top-down fashion. In contrast, in this step, we construct the output matrices at different recursion levels in a bottom-up fashion. In other words, we compute the last stage $S_{h}$ first, then stage $S_{h-1}$, and so on until stage $S_0$. At stage $S_0$, the final output matrix $X$ will be of size $n\times n$. 

Please refer to Figure (the \textsc{Compute-Output-Matrices} algorithm) for computing the output matrix for all leaf nodes in all stages. We first descend the tree until we reach the last stage $S_h$. We know that all cells in the leaf nodes of this stage (or the recursion tree) already store the input $U$ and $V$ matrices using which we can compute the output matrices at that level. Using these output matrices at the leaf level of $S_h$, we compute the output matrices at the root level of $S_h$ (or the leaf level of $S_{h-1}$). Using these matrices at the leaf level of $S_{h-1}$, we compute the output matrices at the root level of $S_{h-1}$. This process continues until we reach the root level of $S_0$ (or the root node of the entire recursion tree), which is the desired matrix product.

The way an output matrix at the root level is computed from the output matrices at the leaf level of stage $S_i$ is as follows. For all cells in the output matrix at the root level, we invoke the function \textsc{Compute-Output-Cell}. This function computes the final value at that cell. 

The working of the \textsc{Compute-Output-Cell} algorithm is shown in Figure \ref{fig:strassen-singlepoint-computation} (right). The figure shows the way in which a highlighted cell in the output matrix at the root level is computed. The highlighted cell belongs to the first quadrant (11) of the output matrix. From the logic of the Strassen's MM algorithm, we know that the first quadrant of the output matrix of a node is computed by adding the output matrices of the 0th, 3rd, and 6th child nodes and subtracting that of the 4th child node. We can compute the highlighted cell from four cells in the next level in $\Oh{1}$ time. Now, consider the output cell in the 3rd child node of the root node. This cell belongs to the third quadrant (22) of that matrix. Again, from the logic of the Strassen's MM algorithm, the fourth quadrant of the matrix is computed by adding the output matrices of the 0th, 2nd, and 5th child nodes and subtracting that of the 1st child node. We continue the process until we reach the leaf level of that stage.

Once cells in the output matrices at the root level of a stage $S_i$ are computed, the algorithm will proceed to computing the cells in the output matrices at the root level of stage $S_{i-1}$ recursively until we reach the root node of the entire recursion tree.

\begin{lemma}
\label{lem:computing-output-matrices}
\textsc{Compute-Output-Matrices} has a complexity of $\Oh{n^w \log \log n}$ work, $\Oh{\log n}$ span, and $\Oh{n^w}$ space.
\end{lemma}
\begin{proof} $[$Work.$]$ We compute the output matrix in $h+1$ \bands $S_0, S_1, \ldots, S_{h}$, where\\
	$h = \log_{w -1}(\log n/\log \log \log n)$. Suppose $W_i$ defines the work done at \band $S_i$. We first come up with a generic formula for $W_i$. We use a direct proof to show that $W_i = \Oh{n^w}$, which implies that the total work is $\sum_{i = 0}^{h} W_i = \Oh{n^w \log \log n}$.
	
	We compute $W_i$ for $i \in [0, h - 1]$. Each output matrix at the root level of stage $S_i$ is constructed from the output matrices at the leaf level of the stage. All cells in all output matrices at the root level of stage $S_i$ are computed in parallel. \#Nodes at the root level of \band $S_i$ is $7^{d_{i-1}\log n} = n^{w d_{i-1}}$. Each such matrix has $(n^{1-d_{i-1}})^2 = n^{2(1-d_{i-1})}$ cells. Each cell at a recursion level $\ell$ depends on at most $4$ output cells in recursion level $\ell +1$. Hence, each cell in a matrix at the root level of stage $S_i$ depends on $\mathcal{O}(4^{(d_{i}- d_{i-1})\log n}) = \Oh{n^{2(d_{i}- d_{i-1})}}$ cells in a matrix at the leaf level of the stage. Hence, $W_i = \Oh{n^{w d_{i-1}} n^{2(1-d_{i-1})} n^{2(d_{i}- d_{i-1})}} = \Oh{n^{2d_i + (w-4) d_{i-1} + 2}}.$
	
	To show that $W_i = \Oh{n^w}$ for all $i \in [0, h-1]$, it is enough to prove that $2d_i + (w-4) d_{i-1} + 2 = w$. We substitute the values of $d_i$ and $d_{i - 1}$ from equation \ref{eq:stages-bounds-2} and simplify to get: $2d_i + (w-4) d_{i-1} + 2 = w$.
	
	We now compute $W_h$. We see that $W_h = \mathcal{O} (n^{2d_h + (w-4) d_{h-1} + 2})$ using the equation aforementioned. We substitute the values of $d_h$, $d_{h - 1}$, $h$, and $\beta$ from equation \ref{eq:stages-bounds-2} and simplify to obtain $W_h = \mathcal{O}(n^w \cdot n^{(4-w)(1-d_{h-1})}) = \mathcal{O}(n^w \cdot n^{2\beta^{h+1}}) = \Oh{n^w \log \log n}$.
	
	Combining the work of the first $h$ \bands and the last \band, we get $T_1(n) = \sum_{i = 0}^{h-1} W_i + W_h = \Oh{n^w \log\log n}$.
	
	\vspace{0.1cm}
	\noindent
	[Span.] Let $T_{\infty}(m, i)$ denote the span of the \textsc{Compute-Output-Matrices} algorithm starting from stage $i$ where a matrix at the leaf level is of size $m \times m$. We give a recursive formula to compute $T_{\infty}(m, i)$. Then, the total span for the algorithm is $T_{\infty}(n, 0)$.
	
	Consider the \textsc{Compute-Output-Matrices} algorithm.  Let $\Delta d_i = d_{i+1} - d_{i}$.
	\#Nodes at the leaf level of stage $S_i$ is $7^{\Delta d_i \log n}$. Launching these nodes in parallel (line \ref{line:leafnodeslaunching}) incurs a span of $\Oh{\Delta d_i w\log n}$. A matrix at the leaf level will be of size $m/(2^{\Delta d_i\log n}) \times m/(2^{\Delta d_i\log n})$. Spawning \textsc{Compute-Output-Cell} function for all cells (lines \ref{line:spawnoutputrows}, \ref{line:spawnoutputcols}) incur a span of $\Oh{2\log m - 2\Delta d_i\log n}$. Executing the \textsc{Compute-Output-Cell} algorithm incurs $\Oh{\Delta d_i\log n}$ span. Adding all these spans give us $\Oh{\Delta d_i \log n + \log m}$.
	
	\hide{
		We consider function reduction with stage id = $i$. The output matrices at the leaf level of stage $i$ are constructed from the leaf level output matrices of stage $i+1$. The number of nodes at the leaf level of stage $S_{i+1}$ is $7^{(d_{i+1}-c_i)\log n}$. Launching these node in parallel incurs a span of $(d_{i+1}-d_i)w\log n$. The size of a node (matrix) at the leaf level is $m/(2^{(d_{i+1} - d_i)\log n})$ where a matrix at the root level of stage $S_{i+1}$ (or at the leaf level of stage $S_{i}$) is of size $m\times m$. The two for loops incur a span of $(2\log m - 2(d_{i+1}- d_i)\log n)$. Computing the input value for each cell of a matrix incurs $\Oh{(d_{i+1}- d_i)\log n}$ span. 
	}
	
	The span of stage $S_i$ recursively depends upon the span of stage $S_{i + 1}$ as the output matrices at the leaf level of stage $S_{i + 1}$ are constructed from the leaf level matrices of stage $S_{i}$. Hence, $T_{\infty}(m,i)$ can be recursively defined using the previous analysis as:
	$T_{\infty}(m,i) = \Oh{\log \log \log n}$ if $i = h$ and $T_{\infty}(m,i) = \Oh{\Delta d_i \log n + \log m} + T_{\infty}(m^{1 - \Delta d_i},i+1)$ if $i < h$. Substituting the values of $d_i$ from equation \ref{eq:stages-bounds-2}, we get $\Delta d_i = 1 - \beta^{i+2} - (1 - \beta^{i+1}) = \beta^{i+1}(1-\beta) = \Oh{\beta^i}$. We know that $m$ starts with $n$ and decreases by a factor of $n^{\Delta d_i}$ for every stage. Hence, $m = n^{1-d_{i-1}} = n^{\beta^i}$, which implies that $\log m = \beta^i \log n$.
	
	\hide{
		$T_{\infty}(m,i) = \begin{cases} \Th{1} &\mbox{if } i = \log\log n, \\
		(d_{i+1}-d_i)w\log n + (2\log m - 2(d_{i+1}- d_i)\log n) \\ +  \Oh{(d_{i+1}- d_i)\log n} + T_{\infty}(m^{1 - (d_{i+1}- d_i)},i+1) \end{cases}$
	}
	
	By unrolling the recursion and using the fact that $\beta^{i}$ is a geometric series and $\beta < 1$, we compute the total span as $T_{\infty}(n,0) = \sum_{i = 0}^{h-1} \beta^i \log n + T_{\infty}(n,h) = \Oh{\log n}$.\\
	
	\vspace{-0.2cm}
	\noindent
	[Space.] Using a similar analysis as given in Lemma \ref{lem:computing-output-matrices}, space usage is $\Oh{n^w}$.
	
	\hide{
		We now simplify $(d_{i+1}- d_i) = 1 - \beta^{i+2} - (1 - \beta^{i+1}) = \beta^{i+1}(1-\beta)$. As $(1-\beta)$ is a constant, substituting it with $d^{'}$, we get $(d_{i+1}- d_i) = d^{'}\beta^{i+1}$. 
		
		Replacing $(d_{i+1}- d_i)$ with $d^{'}\beta^{i+1}$ in the equation of $T_{\infty}(m,i)$, we get the following equation.
		
		$T_{\infty}(m,i) = \begin{cases} \Th{1} &\mbox{if } i = \log\log n, \\
		d^{''}\beta^{i+1}\log n + T_{\infty}(m^{1 - d_i \beta^{i+1}},i+1) &\mbox{where } d^{''} \mbox{is some constant.}  \end{cases}$
		
		Solving this equation for $T_{\infty}(n,0)$, we get the span of reduction phase as $O(\log n)$ as $\beta^{i}$ is a geometric series and $\beta < 1$.
	}
	
\end{proof}

\begin{theorem}
\label{thm:strassen-s}
The \textsc{Strassen-S} MM algorithm has a complexity of $\Oh{n^w \log \log n}$ work, $\Oh{\log n}$ span, and $\Oh{n^w}$ space.
\end{theorem}
\begin{proof} The theorem follows from lemmas \ref{lem:computing-input-matrices} and \ref{lem:computing-output-matrices}.
\end{proof}

\paragraph{\textsc{Strassen-S-Adaptive} MM.} We design a parallel Strassen's MM algorithm \textsc{Strassen-S-Adaptive} with space-span trade-off, which for any given $s$ amount of space in the range $[n^2, n^w]$, achieves the optimal span for that space and performing work very close to $\Oh{n^w}$. Suppose we are given the input matrices $U$ and $V$. We need to compute the output matrix $X$ using space $s \in [n^2, n^w]$. Then, the algorithm works as follows. Observe that there are $\log n$ levels in the recursion tree of the Strassen's MM algorithm. We split the entire recursion tree, at level $t$, into two parts: the top part and the bottom part. The threshold level $t$ depends on the value $s$. We execute the classical sequential Strassen's MM algorithm in the top portion of the recursion tree and the \textsc{Strassen-S} algorithm in the bottom portion of the tree.

\begin{theorem}
\label{thm:strassen-s-tunable}
The \textsc{Strassen-S-Adaptive} MM algorithm has a complexity of $\Oh{n^w \log\log n}$ work and $\Oh{(n^w/s)\log n}$ span, given $\Th{s}$ amount of space.
\end{theorem}
\begin{proof}
	Let $T_1(n,s)$, $T_{\infty}(n, s)$, and $S_{\infty}(n,s)$ denote work, span, and space of \textsc{Strassen-S-Tunable}. Let $T_1(n)$, $T_{\infty}(n)$, and $S_{\infty}(n)$ denote work, span, and space of \textsc{Strassen-S}. Note that \textsc{Strassen-S} does not have a space parameter.
	
	We run the sequential Strassen's MM algorithm for the first $t$ levels of the recursion tree. At level $t$, the size of the matrices is $n/2^t \times n/2^t$ and the number of matrices is $7^t$. We have
	\begin{align*}
	&T_{1}(n,s) = 7T_{1}(n/2,s) + \Oh{n^2} =  \cdots = 7^t T_1(n/2^t) + \Oh{7^t n^2}\\
	&T_{\infty}(n, s) = 7 T_{\infty}(n/ 2, s) + \Oh{\log n} = \cdots = 7^t T_{\infty}(n/2^t) + \Oh{7^t \log n}\\
	&S_{\infty}(n,s) = S_{\infty}(n/2, s) + \Oh{n^2} = \cdots = S_{\infty}(n/2^t) + \Oh{n^2} = \Oh{(n/2^t)^w + n^2}
	\end{align*}
	
	Equating the total space usage with $s$, we get $s = \Th{(n/2^t)^w + n^2}$. We simplify this expression to get the two expressions $n/2^t = \Th{(s - n^2)^{(1/w)}}$ and $7^t = \Th{n^w/(s - n^2)}$.
	Substituting the two expressions in the span and work equations, we have
	\begin{align*}
	&T_{\infty}(n/2^t) = \Oh{\log^2(n/2^t)} = \Oh{\log^2 s}= \Oh{\log n}\\   
	&T_{\infty}(n, s) = 7^t T_{\infty}(n/2^t) + \Oh{7^t \log n} = \Oh{(n^w/s)\log n}\\
	&T_{1}(n/2^t) = (n/2^t)^w \log\log (n/2^t)\\
	&T_{1}(n,s) = 7^t\cdot T_1(n/2^t) + \Oh{7^t n^2} = \Oh{n^w \log\log n}
	\end{align*}
\end{proof}

\vspace{0.2cm}
\para{\textsc{Strassen-W-Adaptive} MM} We can show that by using the sequential 1-way Strassen's MM until recursion level $t$ (depends on $s$ units of space) and then switching to the 7-way Strassen's MM instead of \textsc{Strassen-S}, we can achieve work bound the same as that of the classical Strassen's MM, but the span increases by an extra $\Oh{ \log n}$ factor.
\begin{theorem}
\label{thm:strassen-w-tunable}
The \textsc{Strassen-W-Adaptive} MM algorithm has a complexity of $\Oh{n^w}$ work and $\Oh{(n^w/s)\log^2 n}$ span, given $\Th{s}$ amount of space.
\end{theorem}
\begin{proof}
The proof is similar to that of Theorem \ref{thm:strassen-s-tunable} except that $T_{\infty}(n/2^t) = \Oh{\log^2 n}$ and $T_{1}(n/2^t) = (n/2^t)^w$. 
\end{proof}

\begin{corollary}
With $s = \Th{n^2}$ units of space, $(i)$ \textsc{Strassen-S-Adaptive} has a complexity of $\Oh{n^w\log\log n}$ work and $\Oh{n^{w-2}\log n}$ span and $(ii)$ \textsc{Strassen-W-Adaptive} has a complexity of $\Oh{n^w}$ work and $\Oh{n^{w-2}\log^2 n}$ span.
\end{corollary}

\vspace{0.2cm}
\para{\textsc{Strassen-W} MM} We ask the following question. If the work is bounded by $\Oh{n^w}$, what is the best span achievable by a parallel Strassen's MM algorithm? It turns out that with $\Oh{n^w}$ work bound one can achieve $\Oh{\log n \log\log\log n}$ span. 

We split the entire recursion tree, at level $\log (n/(\log \log n)^{1/(w-2)})$, into two parts: the top and the bottom parts. We execute the \textsc{Strassen-S} in the top portion and the quadratic space \textsc{Strassen-W-Adaptive} algorithm in the bottom portion. 

\hide{Given an input matrix of size $n\times n$, we divide it into $m^2$ blocks each of size $(n/m)\times(n/m)$ where $m = n/(\log\log n)^{1/(w-2)})$. We run \textsc{Strassen-OPT} on the blocked $m\times m$ matrix where each cell is an $m\times m$ matrix instead of a single numerical value.
At each leaf of \textsc{Strassen-OPT} when we need to multiply two such $(n/m)\times(n/m)$ matrices, we multiply them with quadratic space \textsc{Strassen-(1,7)} algorithm. 
At each internal node of \textsc{Strassen-OPT} we add two $(n/m)\times(n/m)$ matrices in $\log (n/m)$ span with two parallel for loops.}

\begin{theorem}
\label{thm:strassen-w}
The \textsc{Strassen-W} MM algorithm has a complexity of $\Oh{n^w}$ work, $ \Oh{\log n \cdot \log\log \log n}$ span, and $\Oh{n^w/\log \log n}$ space. 
\end{theorem}
\begin{proof} There are $m$ matrices of size $(n/m)\times(n/m)$ at the switching level $t$, where $m = n/(\log\log n)^{1/(w-2)})$. At each node at the threshold level, we add two matrices with two for loops. Adding two matrices has $\Oh{\log (n/m)}$ span and $\Oh{(n/m)^2}$ work. From Theorem \ref{thm:strassen-s}, the span and work for \textsc{Strassen-S} are  $\Oh{\log m \cdot \log (n/m)}$ and $\Oh{(m^w \log\log m) \cdot (n/m)^2}$ respectively. The span and work for \textsc{Strassen-W-Tunable} in $m^w$ leaves are $\Oh{(n/m)^{(w-2)}\log^2 (n/m)}$ and $\Oh{m^w(n/m)^w}$ respectively. Combining the span from \textsc{Strassen-S} and \textsc{Strassen-W-Tunable}, we get the expression for span for the hybrid algorithm as follows.
	
	When $m = n/(\log\log n)^{1/(w-2)}$, then $n/m = (\log\log n)^{1/(w-2)}$. We compute span as
	\begin{align*}
	&T_{\infty}(n) = \Oh{ \log m \cdot \log (n/m)  + (n/m)^{w-2}\log^2 (n/m)}\\
	&= \Oh{( (\log n - (1/(w-2)) \log\log\log n )(1/(w-2))\log\log\log n ) + ((1/(w-2))\log\log\log n)^2 \log\log n}\\
	&=\Oh{\log n \cdot \log\log\log n }.
	\end{align*}
	Combining work from both \textsc{Strassen-S} and \textsc{Strassen-W-Tunable}, we get the total work as:
	\begin{align*}
	T_1(n) &= \Oh{(m^w \log\log m)(n/m)^2 + m^w(n/m)^w}\\
	&= \Oh{(n^w/ (\log\log n)^\frac{w}{w-2})(\log\log n)^{\frac{2}{w-2}} + n^w} = \Oh{n^w}.    
	\end{align*}
	We use $s = \Th{n^w/\log\log n}$ space for the whole algorithm.
	$$
	S_{\infty}(n) = \Oh{m^w (n/m)^2} 
	= \Oh{(n^w/ (\log\log n)^\frac{w}{w-2})(\log\log n)^{\frac{2}{w-2}}} 
	= \Oh{n^w/\log\log n}.
	$$
\end{proof}

\vspace{0.2cm}
\noindent
\textbf{Strassen-like MM Algorithms.}
Let recursive algorithm ALG multiply two input matrices $U$ and $V$ of size $n\times n$ and produce output matrix $X$, that is $X = U \cdot V$. ALG divides $U$ into $m\times m$ blocks each of size $(n/m)\times (n/m)$. Similarly, ALG divides the other input matrix $V$ and output matrix $X$. Suppose that algorithm ALG has $R$ recursive calls in each level of recursion. Then ALG creates $R$ temporary matrices each for both input matrices and the output matrix. In particular, the computation of ALG in each level of recursion is as follows. For each $r = 0,1,\ldots, R-1$, Here $A^{(r)}, B^{(r)}$ and $C^{(r)}$ represent temporary matrices.
$$
\begin{aligned}
& A^{(r)} \longleftarrow 0; B_{(r)} \longleftarrow 0 \\
& A^{(r)} \longleftarrow A^{(r)} + \alpha_{i,k}^{(r)}U_{i,k} \text{ for } i,k = 0,1,\cdots, m-1 \\
& B^{(r)} \longleftarrow B^{(r)} + \beta_{k,j}^{(r)}V_{k,j} \text{ for } k,j = 0,1,\ldots, m-1 \\
& C^{(r)} \longleftarrow A^{(r)} \cdot B^{(r)} \\
& X_{i,j} \longleftarrow  X_{i,j} + \gamma_{i,j}^{(r)} C^{(r)} \text{ for } i,j = 0,1,\ldots, m-1.
\end{aligned}
$$
We call ALG a Strassen-like MM algorithm~\cite{mccoll1999memory}.
Sequential algorithm ALG does $\Oh{n^w}$ work where $w = \log_m R$. In Strassen, $m = 2$ and $R = 7$.

It is important to observe that the approaches used by \textsc{Strassen-S}, \textsc{Strassen-W}, \textsc{Strassen-S-Adaptive}, and \textsc{Strassen-W-Adaptive} MM algorithms apply to all Strassen-like MM algorithms.

\paragraph{Lower Bounds.} We give the following lower bound of any parallel version of Strassen's MM algorithm using $s$ units of space in the binary-forking model.

\begin{theorem}
\label{thm:strassen-lower-bound}
Let $A$ be a parallel version of Strassen's MM algorithm which uses $s$ units of extra space. Then $A$'s span is $\Om{\max(n^w / s, \log n)}$ in the binary-forking model.
\end{theorem}

\begin{proof}
We consider the binary-forking model without atomics~\cite{BlellochFiGuSu2019}. In this model, every binary operation (addition, subtraction, multiplication and division) is associated with a memory location. Specifically, the output of such binary operation needs to be written to a memory location. Concurrent reads from a memory location are allowed, while concurrent writes are not.

Let $A$ be a matrix multiplication serial algorithm with $\Th{n^w}$ work where $w \ge 2$. Let $T_{\infty}(n, w, s)$ denote the span of any parallel algorithm $B$ that parallelize the serial algorithm $A$ using $s$ units of space. We do not make any restriction on the number of processors used. We remark that only heap space is counted in $s$ (all our previous algorithms also allocate space from heap memory). 

We get the first lower bound as follows. When we use $s$ units of memory locations, then from the pigeonhole principle, there must exist a memory location that is subjected to $\Th{n^w/s}$ write operations. As concurrent writes to a memory location are not allowed, the following lower bound holds: $T_{\infty}(n, w, s) = \Om{n^w/s}.$

We get the second lower bound from a more general PRAM CREW model. Thus, it holds for a more restricted binary-forking model. When we have an unbounded \#memory locations and unbounded \#processors, the following lower bound ~\cite{jaja1997introduction} holds for the computation of $x_1+x_2+\cdots + x_n$ (array-sum) where every $s_i\in \{0,1\}$: $T_{\infty}(n, w, \infty) = \Om{\log n}.$

As matrix multiplication is at least as hard as array-sum, it must have $\Omega(\log n)$ span. Combining both the lower bounds, we get the following: $T_{\infty}(n, w, s) = \Om{\max(n^w / s, \log n)}.$
\end{proof}

\section{Sorting}
\label{sec:sorting}

Sorting a set of $n$ numbers is one of the most fundamental problems in computer science. Several efficient sorting algorithms exist \cite{Knuth1998, jaja1997introduction, Akl2014, EstivillWo1992, Blelloch1998, Hirschberg1978}. In this section, we consider work-optimal (i.e., performing $\Oh{n \log n}$ work) comparison sort algorithms in the binary-forking model. Our goal is to design a work-optimal comparison sort parallel algorithm with optimal span of $\Oh{\log n}$ in the binary-forking model without atomics.

Cole's pipelined merge sort \cite{cole1988parallel} has optimal span of $\Oh{\log n}$ in the PRAM model. However, it achieves a span of $\Oh{\log^2 n}$ when analyzed in the binary-forking model. Blelloch et al.'s randomized sorting algorithm \cite{BlellochGiSi2010} achieves $\Oh{n \log n}$ work and $\Oh{\log^{1.5} n}$ span, both bounds w.h.p. in $n$. Cole and Ramachandran's sample-partition-merge-sort algorithm \cite{ColeRa2017} is based on multi-way merge and sample sort \cite{FrazerMc1970} and has a low span of $\Oh{\log n \log \log n}$. Blelloch et al.'s randomized sorting algorithm \cite{BlellochFiGuSu2019} is based on sample sort and has $\Oh{n \log n}$ expected work and $\Oh{\log n}$ span w.h.p. in $n$, but it makes use of the test-and-set instruction and as a result cannot be straightforwardly adapted to the binary-forking model without atomics.

\paragraph{Avoiding TS and Retaining Expected Work.}
Our primary strategy for avoiding atomic TS operations is to use extra space and randomize write locations, rendering collisions unlikely. In Blelloch et al.'s \cite{BlellochFiGuSu2019} randomized sorting algorithm, $n$ elements are partitioned by repeatedly attempting to place them into random locations in $\Th{n}$. Since there is a $\Th{1}$ likelihood of success on every write attempt, it is enough for each item to try $\Th{\log n}$ times to find an unoccupied cell with high probability; the span of this process is $\Th{\log n}$. This method cannot be applied in $\Th{\log n}$ span without TS.

Instead of using test-and-set, we use $2n\log n$ extra space. Each element randomly picks $\log n$ cells and tries to put itself into all of them in parallel. If multiple elements try to put themselves into the same cell, they collide and only one succeeds. For each element $e$ set a memory location $f_e$ that holds the cell where $e$ should go. After the attempt is completed, if a cell at memory location $i$ holds element $e$, then it tries to write its index $i$ into $f_e$. If element $e$ succeeds in putting itself into at least one of the $\log n$ randomly chosen cells, then $f_e$ will w.h.p. in $n$ hold an index where element $e$ is stored. After all memory locations $f_e$ are populated, compact the elements into contiguous locations using the content of the memory locations.

In the algorithm by Blelloch et al., every element tries to put itself into a memory location $\log n$ times serially where the success probability of each attempt is $1/2$. In our algorithm, each element tries in parallel $\log n$ times where the success probability of each attempt is $1/2$. Hence, the work remains the same, that is $\Th{n\log n}$. The span is also the same, the is $\Th{\log n}$.
We increase the success probability for an item to land in an unoccupied place by using $\Th{\log n}$ times extra space instead of trying $\log n$ times.

\paragraph{Improving Work Bound from Expectation to High Probability.} 
Here we present a sorting algorithm which takes optimal work and span, both with high probability in $n$. We achieve this by first sorting all but $\oh{n}$ elements of an array in optimal work and span, w.h.p. in $n$ and later integrating the leftover $\oh{n}$ elements with the sorted $n - \oh{n}$ elements. We also give a space-adaptive sorting algorithm that achieves near-optimal work and span w.h.p. in $n$ for any given amount of space $s$.

First we show that all but $\oh{n}$ elements can be sorted in optimal work and span w.h.p. in $n$ by using {\AlmostSort}, a very shallow divide and conquer algorithm of similar design as Blelloch et al.'s, in that it recursively partitions elements into buckets. Although Blelloch et al.'s algorithm gives high probability at all levels of recursion, it is only high probability in the size of arrays passed to the recursive steps, which becomes subpolynomial in $n$ (the original input size) at the lowest levels of recursion. This causes the related work bound to hold only in expectation. In {\AlmostSort} we only use recursive partitioning to depth $\log \log \log n$, after which we apply the Cole-Ramachandran deterministic sorting algorithm. This keeps the array size large enough to achieve a high probability bound on work at all levels of recursion.

Partitioning an array of size $n$ in {\AlmostSort} is done in two steps. First we sample a set of $\sqrt{n}\log^3 n$ potential pivots from the array, sort this subset using {\EpsilonWaySort} (see Figure \ref{fig:multiway-sort}), and select $\sqrt{n}$ uniformly spaced elements from it to serve as pivots. Second we attempt to place every element in the array into two randomly chosen locations in their respective buckets. If two elements attempt to write to the same location, then both attempts fail. If both of a single element's attempts to be written into its partition fail, then that element is dropped and will not appear in the next recursive step. The pseudocode listing for {\AlmostSort} is given in figure \ref{fig:sorting-algorithm-almost}.

We then give a method for recombining the $\oh{n}$ unsorted elements (those that were dropped by {\AlmostSort}) into the sorted array in $\Oh{n \log n}$ work and $\Oh{\log n}$ span, w.h.p. in $n$. This method also depends on using extra space to convert the problem of merging a small unsorted array with a large sorted array into a series of bins-and-balls type scenarios. The union of {\AlmostSort} with this integration method gives us the {\MainSort} algorithm, which w.h.p. in $n$ sorts an array of size $n$ in optimal $\Oh{n \log n}$ work and $\Th{\log n}$ span, using a $\Th{\frac{\log n \log \log \log n}{\log \log n}}$ factor extra space.

Our main results are given in theorem \ref{thm:sorting-almost-sort} and theorem  \ref{thm:sorting-main-sort}.

\begin{lemma}\label{lem:sorting-partitioning}
Partitioning an array of size $n$ into $\sqrt{n}$ blocks with oversampling factor $\log^3 n$ will produce no blocks with size falling outside of an $\varepsilon = 1 / \log^2 \log n$ factor of $\sqrt{n}$, with probability at least $1 - n^{-\log n}$.
\end{lemma}
\begin{proof} % TODO: use letter other than m
	Let $\langle a_1, ..., a_n \rangle$ be the input array $A$ as a sorted sequence, from which $q = \sqrt{n} - 1$ pivots are chosen (with repetition) with oversampling factor $s = \log^3 n$. If one of the resulting blocks contains at least (resp. less than) $k$ elements, then there must be a subsequence $\langle a_j, ..., a_{j+k-1} \rangle$ from which less than (resp. at least) $s$ elements were sampled. \hide{To restate in the interest of clarity: we will show that there is a polynomially small chance of our generating a partition so irregular that some block contains more than (resp. less than) $\varepsilon$ times more (resp. less) elements than ``expected'', where expectation is naively defined in terms of the number of elements being partitioned and the number of partition blocks being used.}
	
	Let $\{x_1, ..., x_{(q+1)s}\}$ be the elements sampled from $A$, and define the Bernoulli variables $X_i = 1$ if $x_i \in \langle a_j, ..., a_{j+k-1} \rangle$, $X_i = 0$ otherwise, and $X = \sum_{i = 1}^{(q+1)s} X_i$ for some arbitrary fixed $j$. Note that $E(X) = (q+1) s E(X_i) = \frac{s k}{\sqrt{n}}$. Let $b_i$ be the sizes of the blocks of the partition of $A$ with pivots $x_s, x_{2s}, ..., x_{qs}$: we then have $P(\max_i b_i \geq k) \leq P(\exists j \; X < s)$ (resp. $P(\min_i b_i < k) \leq P(\exists j \; X \geq s))$.
	
	We handle the upper tail by setting $k = (1 + \varepsilon)\sqrt{n}$ gives us $E(X) = (1 + \varepsilon)s$, from which we can upper bound the probability of an unexpectedly large partition block by using a Chernoff bound for Poisson Binomial distributions:
	\begin{align*}
	P\left(\max_i b_i \geq (1 + \varepsilon)\sqrt{n}\right) &\leq P(\exists j \; X < s) \\
	&\leq n P\left(X < (1 - \varepsilon/2) E(X)\right) & &\left(1 - \varepsilon/2 > \frac{1}{1 + \varepsilon} \right) \\
	&\leq n \exp\left(-\varepsilon^2 s / 8\right) \\
	&= n \exp\left(\frac{-\log^3 n}{8 \log^4 \log n}\right) \\
	&< (1/2) n^{-\log n} & &(\text{for } n \geq 4).
	\end{align*}
	
	Similar to the treatment of upper tail, for the lower tail we set $k = (1 - \varepsilon)\sqrt{n}$, which gives $E(X) = (1 - \varepsilon)s$. Applying the Chernoff bound for lower tails then shows that
	\begin{align*}
	P\left(\min_i B_i < (1 - \varepsilon)\sqrt{n}\right) &\leq P(\exists j \; X \geq s) \\
	&\leq n P(X > (1 + \varepsilon) E(Y)) & &\left(1 + \varepsilon < \frac{1}{1 - \varepsilon}\right) \\
	&\leq n \exp\left(-\varepsilon^2 (1 - \varepsilon) s / 4\right) \\
	&\leq n \exp\left(\frac{-\log^3 n}{8 \log^4 \log n}\right) \\
	&< (1/2) n^{-\log n} & &(\text{for } n \geq 4).
	\end{align*}
	
	Combining the upper and lower tails immediately gives a bound on the probability of a sufficiently regular partition, $P\left(\forall i \; (1 - \varepsilon)\sqrt{n} \leq b_i \leq (1 + \varepsilon)\sqrt{n}\right) > 1 - n^{-\log n}$ for $n \geq 4$.
\end{proof}

\begin{corollary}\label{cor:sorting-partitioning}
Lemma \ref{lem:sorting-partitioning} continues to hold with high probability in $n$ for arrays of size $m < n$ being partitioned into $\sqrt{m}$ blocks with fixed $\varepsilon = 1/\log^2 \log n$ and oversampling factor $\log^3 n$.
\end{corollary}
\begin{proof} % TODO: use a different letter from m
Decreasing the number of elements to be partitioned while holding the oversampling factor constant can only increase the likelihood of a suitably uniform partition. Therefore, if Lemma \ref{lem:sorting-partitioning} guarantees high probability of a uniform partition for an array of size $n$, then that guarantee will continue to if we lower the number of elements in the array to $m < n$.
\end{proof}

\begin{lemma}\label{lem:sorting-sampling}
Attempting to place $n$ elements in $nm = \Om{\frac{n \log n}{\log \log n}}$ space will with high probability in $n$ take $\Oh{\log n}$ span, $\Th{n}$ work, and result in $\Th{\frac{n}{m}}$ collisions.
\end{lemma}
\begin{proof}
	$[$Collisions.$]$ The likelihood of any given element experiencing a collision is less than $1 / m$, and therefore the fact that $\Th{\frac{n}{m}}$ collisions will occur with high probability follows directly from a pair of of Chernoff bounds, $P(X > 2E(X)) < e^{-E(X) / 3}$ and $P(X < E(X)/2) < e^{-E(X) / 8}$, where we take $X$ to be the number of collisions.
	
	$[$Work and Span.$]$ We assume that $k$ collisions at the same location will produce a span of $\Oh{k}$. Consider the set of $\Th{n/m}$ elements which were not placed due to a collision: they are randomly distributed throughout the $\Th{n}$ locations that were selected by at least one element, so with probability greater than $1 - \exp(-\Th{n/m^2})$ there are $\Th{n/m^2}$ of them that collided with one another. By induction, there are $\Th{n/m^d}$ locations where $d$ elements collided, w.h.p. in $n$ for all $d$ such that $n/m^d = \Om{\log n}$. The remaining $\Oh{\log n}$ locations where $\Om{\log_m n}$ elements collided contribute, at most, $\Oh{\log^2 n}$ work and $\Oh{\log n}$ span.
	
	The span is therefore dominated by the time taken to spawn $n$ processes, $T_\infty(n) = \Th{\log n}$, and work is given by $T_1(n) = \Th{n + n \sum_{d = 1}^{\log_m n} \frac{d^2}{m^d}} = \Th{n}$, where we have used the identity $\sum_{j = 1}^\infty j^2 x^j = \frac{x (1 + x)}{(1 - x)^3} \sim x$ for $x$ small.
	
	\hide{
		, for $n/m^d = \Om{\log n}$. That the remaining $\oh{m \log n}$ items take $\oh{m^2 \log^2 n}$ work and $\oh{\log^2 n}$ span is sufficient to show our bounds.
		
		The total span of this is bounded by $s(n) = \oh{\log^2 n} + \Oh{\log_m n} = \oh{m \log n}$, and the work is bounded with
		$$W(n) = \oh{m^2 \log^2 n} + \Th{n + n \sum_{d = 1}^{\log_m n} \frac{d^2}{m^d}} = \Oh{n \log \log n}.$$
	}
\end{proof}

\begin{figure}
\begin{mycolorbox}{$\AlmostSort(n_c, n, l_0, d, m, B, C, D)$
\vspace{-.3cm}
\begin{align*}
n_c &\text{ : size of array to sort (only $B[l_0, ..., l_0 + n_c - 1]$ is occupied)} \hspace{2cm}~ \\
n &\text{ : size of the array at the highest level of recursion} \\
l_0 &\text{ : location where the array to sort begins} \\
d &\text{ : depth of current call in recursion tree} \\
m &\text{ : multiple of extra memory to use} \\
B[l_0, ..., l_0 + n_c m - 1] &\text{ : contains array to be sorted} \\
C[l_0, ..., l_0 + n_c m - 1] &\text{ : ancillary space} \\
D[l_0, ..., l_0 + n_c m - 1] &\text{ : where prefix sums will be stored (for indexing)}
\end{align*}
}
\begin{minipage}{0.99\textwidth}
{\codesize

\algotopspace{}
\noindent
\begin{enumerate}
\setlength{\itemindent}{-1em}

% ------ base case
\vsitem \xif $d \geq \log \log \log n$ \xdo  $\textsc{Cole-Ramachandran}(B[l_0, l_0 + n_c - 1])$ and \xreturn
\vsitem[] $\{ \text{ Partitioning $B$; choose $\sqrt{n_c} - 1$ pivots with oversampling factor $s$ } \}$ \dotfill
\vsitem $s \gets \log^3 n_c$
\vsitem $P \gets$ sample with repetition $(\sqrt{n_c} + 1) s$ elements from $B[l_0,...,l_0 + n_c - 1]$
\vsitem $P \gets \EpsilonWaySort(P, 1/2)$
\vsitem $P[0] \gets -\infty$; $P[(\sqrt{n_c} + 1) s] \gets +\infty$
\vsitem[] $\{ \text{ Nondeterministically partition $B$ into $C$ } \}$ \dotfill
\vsitem \xparallelfor $a \in B[l_0,...,l_0 + n_c - 1]$ \xdo
\vsitem \T Find some $i$ s.t. $P[i \cdot s] \leq a \leq P[(i+1) s]$ \label{line:implicit-binary-search}
\vsitem \T Choose a random number $j \in [0, ..., m \sqrt{n_c} - 1]$
\vsitem \T Attempt to assign $C[i m \sqrt{n_c} + j] \gets a$; in case of collision do nothing
% ------ cleaning up
\vsitem \xparallelfor $i \gets l_0$ \xto $l_0 + n_c - 1$ \xdo $B[i] \gets $ null
\vsitem[] $\{ \text{ Compact the partitions of $C$ into $B$ } \}$ \dotfill
\vsitem \xparallelfor $i \gets 0$ \xto $\sqrt{n_c} - 1$ \xdo
\vsitem \T lo $\gets l_0 + im\sqrt{n_c}$; hi $\gets l_0 + (i + 1)m\sqrt{n_c} - 1$
\vsitem \T $D[\text{lo}, ..., \text{hi}] \gets \textsc{Indicator-Prefix-Sum}(C[\text{lo}, ..., \text{hi}])$
\vsitem \T \xparallelfor $j \gets$ lo \xto hi \xdo
\vsitem \T \T $D[j] \gets D[j] + im\sqrt{n_c}$
\vsitem \xparallelfor $i \in [l_0,...,l_0 + n_c m - 1]$ \xdo
\vsitem \T \xif $C[i]$ is not null \xdo
\vsitem \T \T $B[D[i]] \gets C[i]$
\vsitem \T \T $C[i] \gets $ null
\vsitem[] $\{ \text{ Divide and Conquer } \}$ \dotfill
\vsitem \xparallelfor $i \gets 0$ \xto $\sqrt{n_c} - 1$ \xdo
\vsitem \T $\AlmostSort(\sqrt{n_c}, n, l_0 + im\sqrt{n_c}, d + 1, m, B, C, D)$
\vsitem[] $\{ \text{ Compact the final result } \}$ \dotfill
\vsitem $C \gets B$; $D \gets \textsc{Indicator-Prefix-Sum}(C)$
\vsitem \xparallelfor $i \gets 0$ \xto $n_c m - 1$ \xdo
\vsitem \T \xif $C[i]$ is not null \xthen
\vsitem \T \T $B[D[i]] \gets C[i]$

\algobottomspace{}
\end{enumerate}
}
\end{minipage}
\end{mycolorbox}
\caption{Listing for the randomized sorting algorithm {\AlmostSort}. $\textsc{Boolean-Prefix-Sum}(A)$ is a standard prefix sum taken over the indicator function of $A$, i.e. the function which is 0 where $A[i]$ is null, and 1 elsewhere. Note that on line \ref{line:implicit-binary-search} we have an implicit binary search across $\sqrt{n_c}$ elements. } 
\label{fig:sorting-algorithm-almost}
\end{figure}

\begin{theorem}\label{thm:sorting-almost-sort}
The {\AlmostSort} algorithm (Figure \ref{fig:sorting-algorithm-almost}) takes $\Oh{n \log n}$ work w.h.p. in $n$, $\Th{\log n}$ span w.h.p. in $n$, and $\Th{\frac{n \log n \log \log \log n}{\log \log n}}$ space to sort all but $\Th{\frac{n \log^2 \log n}{\log^2 n \log \log \log n}}$ elements w.h.p. in $n$ of an array of size $n$.
\end{theorem}
\begin{proof}
	$[$Work.$]$ Let $n$ be the size of the array passed to {\AlmostSort} in the initial call, and $n_c$ be the size of the array passed at some point in the recursion tree. We terminate recursion at depth $\log \log \log n$, at which point the \textsc{Cole-Ramachandran} sorting algorithm is applied, taking $\Oh{n_c \log n_c}$ work. All higher levels of recursion sample $\sqrt{n_c}$ pivots with oversampling factor $\log^3 n$, which are sorted using {\EpsilonWaySort} with $\varepsilon = 1/2$, after which {\AlmostSort} is called on the $\sqrt{n_c}$ partition blocks, each of which will be no larger than $(1 + 1/\log^2 \log n)\sqrt{n_c}$. For every call to {\AlmostSort} there is $\Oh{n_c \frac{\log n \log \log \log n}{\log \log n}}$ work done on prefix sums and $\Oh{n_c \log n_c}$ work done on the binary searches by which elements find which partition block to be placed in. Thus the work is $T_1(n_c, n) = \Oh{n_c \log n_c}$ when depth is greater than $\log \log \log n$, and $T_1(n_c, n) = \Oh{(\sqrt{n_c} \log^3 n)^{3/2} + n_c \frac{\log n \log \log \log n}{\log \log n} + n_c \log n_c} + \sqrt{n_c} T_1((1 + 1/\log^2 \log n)\sqrt{n_c}, n)$ otherwise.
	
	Note that {\AlmostSort} is called $\Oh{n \log \log \log n}$ times, and the likelihood of a suitably uniform partition during any particular call of {\AlmostSort} is bounded below by $1 - n^{-\log n}$ (via corollary \ref{cor:sorting-partitioning}), so all partitions will be suitably uniform with high probability in $n$.
	
	We can upper bound w.h.p. the size of arrays $n_c$ passed to depth $d$ with $(1 + 1/\log^2 \log n)^d n^{2^{-d}} \leq (1 + 1/\log^2 \log n)^{\log_{3/2} \log n} n^{2^{-d}} \lesssim \exp(1 / \log \log n) n^{2^{-d}} \sim (1 + 1 / \log \log n) n^{2^{-d}}$. Therefore in order to be able to apply Lemma \ref{lem:sorting-sampling} to the array partitioning which occurs at depth $d$, we will need $\Om{n^{2^{-d}} \frac{\log n}{\log \log n}}$ space per function instance.
	
	There are $n^{1 - 2^{-d}}$ instances of {\EpsilonWaySort} being called at depth $d$, which cumulatively take $\Oh{\frac{n}{n_c} (\sqrt{n_c}\log^3 n)^{3/2}} = \Oh{n}$ work for\footnote{Note that when $d < \log \log \log n$, we have $n_c^{-1/4} < \left(n^{1 / \log \log n}\right)^{-1/4} < \log^{-c} n$ for all positive constants $c$.} $n_c = \Om{\log^{18} n}$. The size of arrays being processed at the lowest level of recursion is $\Oh{n^{1/\log \log n}}$, so the total work done by the $n/n_c$ calls to \textsc{Cole-Ramachandran} is $\Oh{\frac{n \log n}{\log \log n}}$. Summing across all levels, we have a total of $\Oh{n \log n}$ work (primarily from the binary searches that occur during partitioning).

	$[$Span.$]$ Following a similar line of reasoning as was used for the work bound, we note that {\EpsilonWaySort} (still with $\varepsilon = 1/2$) has span $\Oh{\log n_c}$, and \textsc{Cole-Ramachandran} has span $\Oh{\log n_c \log \log n_c}$, so the recurrence relation for span is
	$$T_\infty(n_c, n) \leq \begin{cases}
	\Oh{\log n_c \log \log n_c} & \text{if depth} \leq \log \log \log n, \\
	\Oh{\log (n_c \log^3 n)} + T_\infty((1 + 1 / \log^2 \log n)\sqrt{n_c}, n) & \text{else},\end{cases}$$
	which, on substituting $n_c = n^{1/\log \log n}$ at lowest depth, is solved by $T_\infty(n, n) = \Oh{\log n}$.
	
	\hide{
		Let $T_\infty(n, n_0)$ be the span of \textsc{Partition}(...), with the same arguments as $W(n, n_0)$. It takes $\Oh{\log ps}$ span to sort $ps$ elements, and the depth of recursion is $\Oh{\log \log n}$, but at every level deeper $\log ps$ is halved (since $p$ is changed to $\sqrt{p}$), so the total span is $T_\infty(n, n_0) = \Oh{(1 + 1/2 + 1/4 + ...)\log ps} = \Oh{\log n_0}$.
	}
	
	$[$Space.$]$ {\AlmostSort} uses $\Th{\frac{n \log n \log \log \log n}{\log \log n}}$ space, by construction. All partitioning for a node is done within a chunk of memory inherited from the parent which instantiated it, and the root node begins with $\Th{\frac{n \log n \log \log \log n}{\log \log n}}$ memory.
	
	$[$Unsorted Elements.$]$ At recursive depth $d$ there will be $n^{1 - 2^{-d}}$ arrays of size $\Th{n^{2^{-d}}}$ from which we are selecting pivots then partitioning. By Lemma \ref{lem:sorting-sampling} this will produce $\Th{\frac{n \log \log n}{\log n \log \log \log n}}$ collisions, and attempting to place elements twice each will lower this number to $\Th{\frac{n \log^2 \log n}{\log^2 n \log^2 \log \log n}}$ elements which failed to find a place in their respective partitions. Summing across all levels of recursion then gives a total of $\Th{\frac{n \log^2 \log n}{\log^2 n \log \log \log n}}$ elements left unsorted due to collisions.
\end{proof}

\begin{lemma}\label{lem:sorting-sequences}
With high probability in $n$, there are no sequences of length $\Th{\frac{\log n}{\log \log n}}$ in the array $A$ that have no elements which appear in $B = \AlmostSort(A, n)$.
\end{lemma}
\begin{proof}
	We know with high probability that the bins at all stages of {\AlmostSort}'s divide and conquer process are filled no more than within a $(1 + 1 / \log \log n)$ factor above or below expectation (see theorem \ref{thm:sorting-almost-sort}), and this gives us an upper bound on the likelihood that any given element $a_j \in A$ fails to be included in $B$:
	$$P(a_j \notin B) \leq 1 - \left(1 - \frac{1 + 1 / \log \log n}{\frac{\log^2 n \log^2 \log \log n}{\log^2 \log n}}\right)^{\log \log \log n} \sim \frac{\log^2 \log n}{\log^2 n \log \log \log n} \equiv \delta.$$
	This result also follows from our high probability bound on the number of collisions which occur during {\AlmostSort}.
	
	Now we can repeat the argument made for Lemma \ref{lem:sorting-partitioning}, interpreting $B$ as having been uniformly randomly sampled\footnote{Using $P(a_j \notin B) \leq \delta$ allows us to unconditionally upper bound the chance that none of a sequence of elements will end up in $B$.} from $A$, and thereby bounding the likelihood that a subsequence of length $\ell$ exists in $A$ from which no elements are included in $B$.
	
	We define the indicator variables $Y_i$ to be $1$ if $b_i \in \langle a_j, ..., a_{j+\ell-1} \rangle$ and $0$ otherwise, and $Y \equiv \sum_i Y_i$, from which we have $P(Y = 0) \sim \prod_i P(a_{j+i} \notin B) \leq \delta^\ell$. The probability of a subsequence of length $\ell$ being contained in $A$ but completely absent from $B$ is then bounded with $P(\exists j \; Y = 0) \leq n P(Y = 0) \leq n \delta^\ell$, so we find that there is a polynomially small chance of there being a subsequence of length $\ell = \frac{\log n}{\log 1/\delta} = \Th{\frac{\log n}{\log \log n + \log \log \log \log n - \log \log \log n}}$ which is entirely absent from $B$.
\end{proof}

\paragraph{Recombining Unsorted Elements.} We now present a method for merging the unsorted elements $A \setminus B$ with the sorted subarray $B = \AlmostSort(A)$, using a sequence of balls-and-bins type arguments to achieve high probability that no element remains unsorted. The rough outline is as follows: (1) Partition $A \setminus B$ by the elements in $B$, and sample elements into their respective partition blocks (found via a binary search through $B$). Use this to index the nonempty partition block. (2) Partition $A \setminus B$ again by the elements in $B$, this time distributing space only to the blocks that are known to be nonempty. Use the result to estimate the exact size of each partition block. (3) Partition $A$ by the elements in $B$, now distributing memory so that every block has $\Om{\frac{\log n}{\log \log n}}$ times more space than needed to compactly store the elements that are bound for it. Sample elements of $A$ into their respective partitions, making $\log n$ parallel placement attempts for each element. (4) Remove all but one copy of each element from their associated partition blocks, sort each partition block, and use a prefix sum to compact the entire array. With high probability the result is a sorted copy of the entire array $A$.

\begin{lemma}\label{lem:sorting-recombining}
The unsorted set of elements $A \setminus B$ can with high probability be merged with the sorted subarray $B = \AlmostSort(A)$ in $\Oh{n \log n}$ work, $\Th{\log n}$ span, and $\Th{\frac{n \log n \log \log \log n}{\log \log n}}$ space.
\end{lemma} 
\begin{proof}
	$[$Step 1.$]$ We start by associating every element of $B$ with a space of size $\Th{\frac{\log n \log \log \log n}{\log \log n}}$, and then place elements from $A \setminus B$ into a single random location in their respective buckets, where collisions of $k$ elements in the same location result in $\Oh{k}$ span and no element being successfully placed. This will take $\Oh{n}$ work and $\Th{\log n}$ span, even if the maximum number of collisions occur (which are bounded by lemma \ref{lem:sorting-sequences}). We then run $\Th{n}$ parallel prefix sums to find the number of elements in each of the $\Th{n}$ partition blocks, taking $\Th{\frac{n \log n \log \log \log n}{\log \log n}}$ work, after which we index the nonempty blocks by way of a prefix sum taking $\Th{n}$ work. Partition blocks with no associated elements do not need any space devoted to them, so the number of buckets we have to consider has been reduced from $\Th{n}$ to $\Oh{\frac{n \log^2 \log n}{\log^2 n \log \log \log n}}$.
	
	\vspace{0.1cm}
	\noindent
	$[$Step 2.$]$ We give each nonempty partition block a space of size $\Th{\frac{\log^3 n \log^2 \log \log n}{\log^3 \log n}}$ and again place every element in $A \setminus B$ into a single location in its associated bucket. By Lemma \ref{lem:sorting-sampling} the number of elements in every bucket will be within a $\Th{\frac{\log^2 \log n}{\log^2 n \log^2 \log \log n}}$ factor of its maximum value. Running a prefix sum through each bucket gives the total number of elements associated with it to within a $\Th{\frac{\log^2 \log n}{\log^2 n \log^2 \log \log n}}$ factor; this cumulatively takes $\Th{\frac{n \log n \log \log \log n}{\log \log n}}$ work. A prefix sum over the $\Th{n}$ individual bucket sizes then gives an approximation of the bounds of every bucket in the final sorted array.
	
	\vspace{0.1cm}
	\noindent
	$[$Step 3.$]$ Using the bounds found in step 2, we give each bucket $\Th{\frac{\log^3 \log n}{\log^3 n \log^2 \log \log n}}$ times more space than the number of elements that will be going into it. Now we can sample every unsorted item into its associated partition block $\log n$ times, as we have $m = \Th{\frac{\log^3 \log n}{\log^2 n \log^2 \log \log n}}$ times extra space for Lemma \ref{lem:sorting-sampling}, which is sufficient for a high probability bound. No more than $1$ in $\Th{\frac{\log^3 \log n}{\log^2 n \log^2 \log \log n}}$ write attempts will produce a collision, so we have with probability greater than $1 - \frac{n}{(\log n)^{\log n}}$ that every unsorted element will be placed into its bucket at least once.
	
	\vspace{0.1cm}
	\noindent
	$[$Step 4.$]$ All duplicates items in each partition are found by using $\Th{n}$ prefix-sum-like methods to search over the $\log n$ locations where each element $a \in A$ was placed. We run $\frac{n \log \log \log n}{\log \log n}$ of these methods in parallel at a time (so there are $\Th{\frac{\log \log n}{\log \log \log  n}}$ chunks of them run serially) in order to stay within our space bound, taking $\Oh{n \log n}$ work and $\Th{\frac{\log^2 \log n}{\log \log \log n}}$ span. After removing duplicate elements we index all nonempty memory locations by running a prefix sum across the entirety of the $\Th{\frac{n \log n \log \log \log n}{\log \log n}}$ space being used, and then use those indices to compact down all elements into our final sorted array.
\end{proof}

\hide{
We have two important pieces of information about the result of {\AlmostSort}: there are $\Th{\frac{n \log^2 \log n}{\log^2 n \log \log \log n}}$ elements from the initial array $A$ which have not been sorted, and there are no continuous subsequences of length $\Th{\frac{\log n}{\log \log n}}$ which are entirely absent from the sorted subarray $B$.
}

\hide{
We allocate a memory block for every element of $B$ and attempt to place elements of $A\setminus B$ into $\log n$ random locations in their associated memory block (found via a binary search through $B$). However, we cannot do this directly, since we only have $\Th{n \log n \log \log \log n / \log \log n}$ memory locations available, or $\Th{ \log n \log \log \log n / \log \log n}$ per element in $B$, which means that if we attempt to apply Lemma \ref{lem:sorting-sampling} directly then our value for $m$ can go almost as low as $\frac{\log \log n}{\log n \log \log \log n}$ (if we hit a subsequence of unsorted elements of length just under $\frac{\log n}{\log \log n}$).
}

\begin{figure}
\centering
\scalebox{0.75}{
\begin{mycolorbox}{$\MainSort(A, n)$}
\begin{minipage}{0.99\textwidth}
{\codesize

\algotopspace{}
\noindent
\begin{enumerate}
\setlength{\itemindent}{-1em}

\vsitem[] $\{ \textbf{ Part 1:}\text{ Sorting the majority of $A$} \}$ \dotfill
\vsitem $m \gets \log n \log \log \log n / \log \log n$; Allocate arrays $B, C, D$ of size $n m$;  $B[0, ..., n-1] \gets A$
\vsitem $\AlmostSort(n, n, 0, 0, m, B, C, D)$
\vsitem num\_sorted $\gets $ smallest $i$ s.t. $B[i] =$ null
% ------ uniform memory; place all of (A - B) once
\vsitem[] $\{ \textbf{ Part 2:}\text{ Combining unsorted elements into $B$} \}$ \dotfill
\vsitem Set all elements of $C$ and $D$ to null
\vsitem Allocate arrays $E, F$ of size $n$
\vsitem[] $\{ \text{ Pass 1: Finding and indexing all nonempty buckets } \}$ \dotfill
\vsitem \xparallelfor $a \in A[0, ..., n-1]$ \xdo
\vsitem \T Find smallest $i$ s.t. $B[i] \leq a < B[i + 1]$
\vsitem \T \xif $a \neq B[i]$ \xthen
\vsitem \T \T Choose a random number $j \in [0, ..., m - 1]$
\vsitem \T \T Attempt to assign $C[i \cdot m + j] \gets a$; in case of collision do nothing
% ------ find all nonempty buckets and index them
\vsitem $D \gets \textsc{Indicator-Prefix-Sum}(C)$
\vsitem \xparallelfor $i \gets 0$ \xto $n$ \xdo
\vsitem \T \xif $D[i \cdot \text{block\_size}] \neq D[(i + 1) \cdot \text{block\_size} - 1]$ \xthen $E[i] = 1$
\vsitem $F \gets \textsc{Prefix-Sum}(E)$
% ------ clear our work space
\vsitem Set all elements of $C$ and $D$ to null
% ------ uniform memory to nonempty buckets; place all of (A - B) once
\vsitem[] $\{ \text{ Pass 2: Approximating the size of each bucket } \}$ \dotfill
\vsitem block\_size $\gets \log^3 n \log^2 \log \log n / \log^3 \log n$
\vsitem \xparallelfor $a \in A[0, ..., n-1]$ \xdo
\vsitem \T Find smallest $i$ s.t. $B[i] \leq a < B[i + 1]$
\vsitem \T \xif $a \neq B[i]$ \xand $E[i] = 1$ \xdo
\vsitem \T \T Choose a random number $j \in [0, ..., \text{block\_size} - 1]$
\vsitem \T \T Attempt to assign $C[(F[i] - 1) \cdot \text{block\_size} + j] \gets a$; in case of collision do nothing
% ------ generate running count approximation
\vsitem $D \gets \textsc{Indicator-Prefix-Sum}(C)$
\vsitem \xparallelfor $i \gets 0$ \xto num\_sorted $-\;1$ \xdo
\vsitem \T \xif $E[i] = 0$ \xthen $E[i] = 1$
\vsitem \T \xelse $E[i] = D[F[i] \cdot \text{block\_size} - 1] - D[(F[i] - 1) \cdot \text{block\_size}]$
\vsitem $F \gets \textsc{Prefix-Sum}(E)$
% ------ clear our work space
\vsitem clear $C$ and $D$
% ------ nonuniform memory; log n placements
\vsitem[] $\{ \text{ Pass 3: Sampling elements into their buckets with $\log n$ repetitions } \}$ \dotfill
\vsitem \xparallelfor $i \gets 0$ \xto $n - 1$ \xdo
\vsitem \T \xparallelfor $j \gets 0$ \xto $\log n - 1$ \xdo
\vsitem \T \T Choose a random number $k = H(j, (F[i] - E[i])m, F[i]m - 1)$ \label{line:hash-function}
\vsitem \T \T Attempt to assign $C[k] \gets a$; in case of collision do nothing
% ------ remove duplicates
\vsitem[] $\{ \text{ Removing duplicates and compacting } \}$ \dotfill
\vsitem chunk\_size $\gets n \log \log \log n / \log \log n$
\vsitem \xfor $i \gets 0$ to $n / \text{chunk\_size} - 1$
\vsitem \T \xparallelfor $a \in A[i \cdot \text{chunk\_size}, ..., (i + 1) \cdot \text{chunk\_size}-1]$ \xdo
\vsitem \T \T Find smallest $i$ s.t. $B[i] \leq a < B[i + 1]$
\vsitem \T \T \T $\textsc{Keep-Single}(C, a, n, H, (F[i] - E[i])m, F[i]m - 1)$ \label{line:keep-single}
\vsitem \xparallelfor $i \gets 0$ \xto num\_sorted $-\;1$ \xdo
\vsitem \T $D[(F[i] - E[i])m, ..., F[i]m - 1] \gets \textsc{Indicator-Prefix-Sum}(C[(F[i] - E[i])m, ..., F[i]m - 1])$
\vsitem \T \xparallelfor $j \gets (F[i] - E[i])m$ \xto $F[i]m - 1$ \xdo
\vsitem \T \T $D[j] \gets D[j] + (F[i] - E[i])m$
\vsitem \xparallelfor $i \gets 0$ \xto $n m-1$ \xdo
\vsitem \T \xif $C[i] \neq $ null \xthen $\{~ A[D[i]] \gets C[i]; C[i] \gets \text{null} ~\}$
\vsitem \xparallelfor $i \gets 0$ \xto num\_sorted $-\;1$ \xdo
\vsitem \T \text{lo} $\gets (F[i] - E[i])m$; \text{hi} $\gets \text{lo} + E[i] - 1$
\vsitem \T $A[\text{lo}, ..., \text{hi}] \gets \textsc{Cole-Ramachandran}(A[\text{lo}, ..., \text{hi}])$
\vsitem $D \gets \textsc{Indicator-Prefix-Sum}(A)$
\vsitem \xparallelfor $i \gets 0$ \xto $n m-1$ \xdo
\vsitem \T \xif $A[i] \neq $ null \xthen $\{~ C[D[i]] \gets A[i]; A[i] \gets \text{null} ~\}$
\vsitem $A \gets C$

\algobottomspace{}
\end{enumerate}
}
\end{minipage}
\end{mycolorbox}
}
\caption{Listing for {\MainSort}. $H(i, \text{lo}, \text{hi})$ on line \ref{line:hash-function} is a hash function which maps the triple $(i, \text{lo}, \text{hi})$ to a single value in the range $[\text{lo}, ..., \text{hi}]$. $\textsc{Keep-Single}(A, a, n, H, \text{lo}, \text{hi})$ on line \ref{line:keep-single} is a function which performs a $\textsc{Prefix-Sum}$-like operation to index the instances of $a$ at the (assumed distinct) locations $H(i, \text{lo}, \text{hi})$ ($i = 0, ..., \log n - 1$), and then sets all but one of the instances to null. }
\label{fig:sorting-algorithm-main}
\end{figure}

\begin{theorem}\label{thm:sorting-main-sort}
The {\MainSort} algorithm (Figure \ref{fig:sorting-algorithm-main}) sorts an array of size $n$ in $\Oh{n \log n}$ work w.h.p., $\Th{\log n}$ span w.h.p., and $\Th{\frac{n \log n \log \log \log n}{\log \log n}}$ space.
\end{theorem}
\begin{proof}
Directly follows from theorem \ref{thm:sorting-almost-sort} and lemma \ref{lem:sorting-recombining}.
\end{proof}

\paragraph{Space-Adaptive Sorting Algorithm.} In the previous paragraphs, we present a $\Th{\log n}$ span parallel sorting algorithm in the binary-forking model without atomic TS using $n\cdot x$ extra space, where $x = \Th{\log n\log\log\log n/\log\log n}$. We now present a space-adaptive sorting algorithm which, given $s \in [n, n \cdot x]$ space, achieves near-optimal span. 

\begin{enumerate}[nolistsep]
\item We divide the input into $m$ equal-sized segments $A_i$ such that $(n/m) x = s$.
\item For each segment $A_i$ of size $n/m$, use {\MainSort} with $s$ amount of extra space.
\item Merge the $m$ segments pairwise recursively.
\end{enumerate}

\begin{theorem}\label{thm:sorting-tunable-sort}
The {\TunableSort} algorithm sorts an array of size $n$ in $\Oh{n \log n}$ work w.h.p. in $n$, $\Oh{(n/s) \log^2 n}$ span w.h.p. in $n$, using $\Th{s}$ amount of space.
\end{theorem}
\begin{proof}
	Sorting each segment takes $\Th{\log (n/m)}$ span and $\Th{(n/m)\log n}$ work w.h.p. in $n$. As sorting each segment consumes the entire extra space $s$, these $m$ calls to {\MainSort} are made serially. Hence, these $m$ sorting steps take $\Th{m \cdot \log (n/m)} = \Th{(nx/s)\log n}$ span w.h.p. in $n$. Merging $m$ sorted segments takes the following span.
	$$
	\sum_{i = 1}^{\log m}\Th{\log \frac{m}{2^i}+ \log{2^i n}{m}} = \sum_{i = 1}^{\log m}\Th{\log m + \log \frac{n}{m}} = \Th{\log n\log m}.
	$$
	The span of {\TunableSort} is $\Oh{n\log^2 n/s}$ w.h.p. in $n$. 
	The work needed to merge $m$ sorted segments is $\Th{n\log\log n}$ w.h.p. in $n$. Hence, the overall work of {\TunableSort} is $\Th{n\log n}$ w.h.p. in $n$.   
\end{proof}

\section{Fast Fourier Transform}
\label{sec:fft}

The Discrete Fourier Transform (DFT) of an array $a$ of $n$ complex numbers is the array $y$ computed as $y[i] = \sum_{j=0}^{n-1} a[j] w_n^{-ij}$ for all $i \in [0, n-1]$, where $w_n = e^{2 \pi \sqrt{-1}/n}$ is a primitive $n$th root of unity. A Fast Fourier Transform (FFT) is an algorithm that computes the DFT of an array rapidly. FFT is often considered as one of the most important algorithms of the 20th century. It is extensively used in digital signal processing. Several FFT algorithms have been designed that perform $\Oh{n \log n}$ work. For example, prime-factor algorithm (or Good-Thomas' FFT) \cite{Good1958, Thomas1963}, Bruun's FFT \cite{Bruun1978}, and Winograd's FFT \cite{Winograd1978}. Designing an FFT algorithm with $\oh{n \log n}$ work is an open problem. 

Consider the recursive divide-and-conquer Cooley-Tukey FFT algorithm \cite{CooleyTu65} and its special case, the radix-2 FFT algorithm \cite{CormenLeRiSt2009}. A straightforward parallelization of the generic Cooley-Tukey algorithm \cite{CooleyTu65} has a complexity of $\Oh{n \log n}$ work (and space) and $\Oh{\log n \log \log n}$ span. A simple parallelization of the radix-2 algorithm \cite{CormenLeRiSt2009} has $\Oh{\log^2 n}$ span. In this section, we aim to design a parallel FFT algorithm with close to optimal span for the binary-forking model without atomics keeping work as closely as possible to $\Oh{n \log n}$. To this end, we first design a simple single-point FFT algorithm that can be used to compute a single entry of the DFT. We then carefully combine an efficient variant of this algorithm with the radix-2 FFT and mixed-radix FFT.

\paragraph{$n^{\phi}$-way FFT (Cooley-Tukey Algorithm, \cite{CooleyTu65}).} The $n^{\phi}$-way FFT algorithm, for $\phi \in [1/\log n, 1/2]$, is defined as follows. We view the $n$-sized array as a $n^{1 - \phi} \times n^{\phi}$ matrix, as shown in Figure \ref{fig:fft-kway-diagram}. In the first phase, we compute FFT of each of the columns recursively in parallel. We then multiply all entries of the matrix by appropriate twiddle factors. In the second phase, we compute FFT of each of the rows recursively in parallel. Finally, the resultant matrix read in the column-major order is the required DFT. 

The work and the span recurrences for the $n^{\phi}$-way FFT algorithm are as follows. Suppose $c \ge 1$ is a fixed constant. If $n \le c$, then $T_{1}(n) = \Oh{1}$ and $T_{\infty}(n) = \Oh{1}$. If $n > c$, then
\begin{align*}
&T_{1}(n) = n^{\phi} T_{1}(n^{1-\phi}) + n^{1-\phi} T_{1}(n^{\phi}) + \Oh{n},\\
&T_{\infty}(n) = T_{\infty}(n^{1 - \phi}) + T_{\infty}(n^{\phi}) +  \Oh{\log n}.
\end{align*}
The work performed by this algorithm is $\Oh{n \log n}$ for all values of $\phi$. If we set $\phi = 1/\log n$, we obtain the 2-way FFT algorithm \cite{CooleyTu65, CormenLeRiSt2009}, with span $\Oh{\log^2 n}$. On the other hand, if we set $\phi = 1/2$, we get the $\sqrt{n}$-way FFT algorithm, a special case of the algorithm given in \cite{CooleyTu65}, with span $\Oh{\log n \log \log n}$. To the best of our knowledge, the $\sqrt{n}$-way FFT algorithm achieves the best span.

\begin{figure}
\centering
\includegraphics[width=0.85\textwidth]{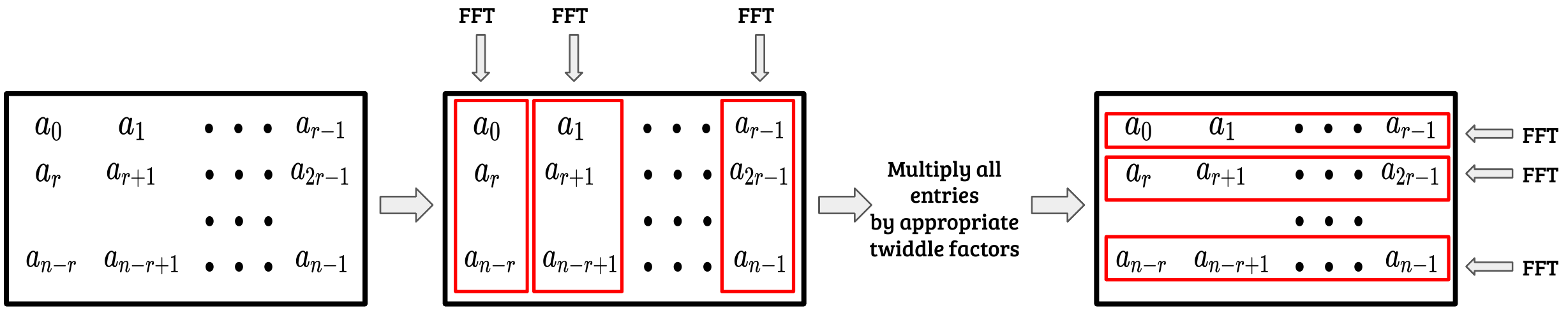}
\caption{The $n^{\phi}$-way FFT algorithm. Here, $r = n^{\phi}$.}
\label{fig:fft-kway-diagram}
\end{figure}

\paragraph{Single-point 2-way FFT.} A single-point FFT evaluation means that we can compute a single entry of the DFT independently without computing the entire DFT. In other words, the computation of $y[i]$ does not share work with the computation of $y[j]$ for all $i,j \in [0, n-1]$ and $i \ne j$. The single-point 2-way algorithm \textsc{FFT-sp} is shown in Figure \ref{fig:fft-single-point-algorithm} (left). Figure \ref{fig:fft-single-point-algorithm} (right) shows a visual depiction of how an entry corresponding to the $i$th entry of the DFT is computed at recursion level $\ell$ using entries corresponding to the $i$th entries of the DFT's of the two child nodes at recursion level $\ell + 1$, for $\ell \in [0, \log n)$.  

\begin{figure}
\begin{minipage}{0.48\textwidth}
\begin{mycolorbox}{$\textsc{FFT-sp}(x, \ell, n, i)$}
\begin{minipage}{0.99\textwidth}
{\codesize

\algotopspace{}
\algoinput Input array $x$, level $\ell$, size $n$, entry $i$
\algooutput DFT entry $y^{\ell}[i]$ 
\algomiddlespace

\noindent
\begin{enumerate}
\setlength{\itemindent}{-1.5em}

\vsitem \xif $n = 1$ \xthen \xreturn $x[0]$
\vsitem $j \gets i - (n/2) \times [i < n/2]$; $w \gets w_n^{-j}$
\vsitem \xpar $u \gets$ \textsc{FFT-sp}$(x,\ell + 1,n/2,j)$
\vsitem[] \xblankpar $v \gets $ \textsc{FFT-sp}$(x + 2^{\ell},\ell + 1,n/2,j)$
\vsitem \xreturn $u - wv \times (-1)^{[i < n/2]}$

\algobottomspace{}
\end{enumerate}
}
\end{minipage}
\end{mycolorbox}
\end{minipage}
\begin{minipage}{0.01\textwidth}
~
\end{minipage}
\begin{minipage}{0.48\textwidth}
\centering
\includegraphics[width=0.99\textwidth]{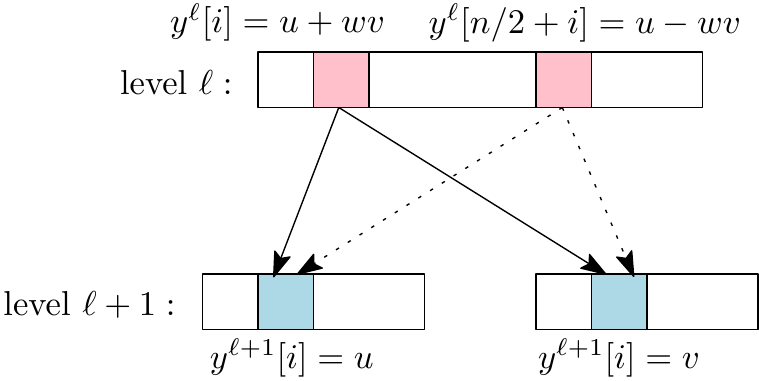}
\end{minipage}
\caption{Left: Single-point FFT algorithm to compute the $i$th entry of the DFT of $n$-sized array $x$. Symbol $[~]$ represents the Iversion bracket. Initial invocation is \textsc{FFT-sp}$(a, 1, n, i)$. Right: Single-point FFT evaluation of DFT entry $y^{\ell}[i]$ at level $\ell$ using the DFT entries $y^{\ell + 1}[i]$ of the child nodes at level $\ell + 1$. Note that $y^{\ell}[n/2+i]$ also depends on the $y^{\ell + 1}[i]$ entries.} 
\label{fig:fft-single-point-algorithm}
\end{figure}

Consider the binary recursion tree produced by the algorithm. It is important to note that an entry in the DFT of a node can be computed with two entries in the DFT's of the child nodes. We do not need to store the DFT entries at every level. They can be computed recursively and on-the-fly. A DFT entry at the root node of the recursion tree can be computed with $\Oh{2^{\log n}} = \Oh{n}$ work and $\Oh{\log n}$ span. This implies that we can compute all entries of the DFT at the root node performing $\Oh{n \cdot 2^{\log n}} = \Oh{n^2}$ work in $\Oh{\log n}$ span.

\begin{lemma}
\label{lem:fft-sp}
The single-point 2-way FFT algorithm has a complexity of $\Oh{n^2}$ work, $\Oh{n}$ space, and $\Oh{\log n}$ span. 
\end{lemma}

\paragraph{2-way FFT with Stages.} The standard 2-way FFT algorithm has a complexity of $\Oh{n \log n}$ work due to high work-sharing across threads and $\Oh{\log^2 n}$ span due to expensive local synchronization points. In contrast, the single-point 2-way FFT algorithm has a complexity of $\Oh{n^2}$ work due to no work-sharing across threads and $\Oh{\log n}$ span due to inexpensive local synchronization points. We now carefully combine the two algorithms to get advantages of both the worlds: low work and low span.

\begin{figure}
\centering
\includegraphics[width=0.35\textwidth]{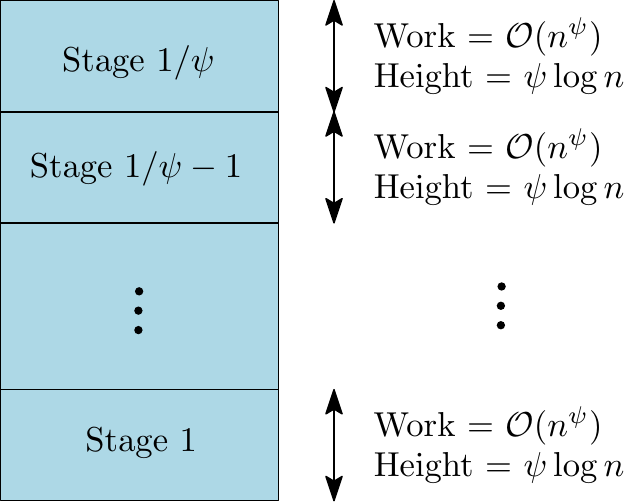}
\caption{Stages in the 2-way FFT algorithm with single-point computations.}
\label{fig:fft-stages-2way}
\end{figure}

Consider Figure \ref{fig:fft-stages-2way}, in which the $\log n$ levels are split into $1/ \psi$ stages, each stage containing $\psi \log n$ levels, for $\psi \in [1/\log n, 1]$. The core idea of the algorithm is as follows. There will be $1/\psi$ global synchronization points, one per stage. We compute the DFT at a root node of stage $i$ using the DFT's at the root nodes of stage $i + 1$ using the single-point 2-way FFT algorithm. We compute and store the DFT's in the root nodes of all stages.

\begin{figure*}
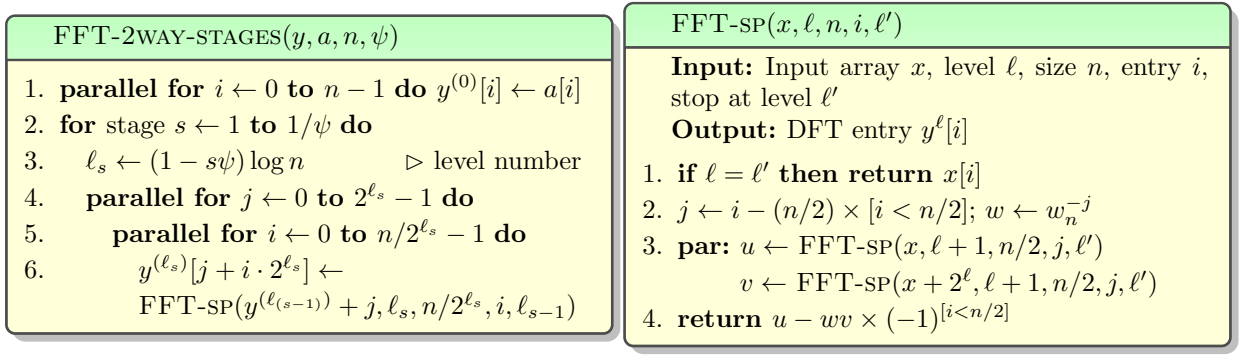

\centering
\begin{minipage}{\textwidth}
\begin{minipage}{0.49\textwidth}
\begin{mycolorbox}{$\textsc{FFT-2way-stages}(y, a, n, \psi)$}
\begin{minipage}{0.99\textwidth}
{\codesize
\algotopspace{}
\noindent
\begin{enumerate}
\setlength{\itemindent}{-1.5em}

\vsitem \xparallelfor $i \gets 0$ \xto $n - 1$ \xdo $y^{(0)}[i] \gets a[i]$

\vsitem \xfor stage $s \gets 1$ \xto $1/\psi$ \xdo
\vsitem \T $\ell_s \gets (1 - s\psi) \log n$ \xcomment level number
\vsitem \T \xparallelfor $j \gets 0$ \xto $2^{\ell_s} - 1$ \xdo
\vsitem \T \T \xparallelfor $i \gets 0$ \xto $n/2^{\ell_s} - 1$ \xdo
\vsitem \T \T \T $y^{(\ell_s)}[j + i \cdot 2^{\ell_s}] \gets$ 
\vsitem[] \T \T \T $\textsc{FFT-sp}(y^{(\ell_{(s-1)})} + j, \ell_s, n/2^{\ell_s}, i, {\ell}_{s - 1})$

\algobottomspace{}
\end{enumerate}
}
\end{minipage}
\end{mycolorbox}
\end{minipage}
\begin{minipage}{0.49\textwidth}
\begin{mycolorbox}{$\textsc{FFT-sp}(x, \ell, n, i, \ell')$}
\begin{minipage}{0.99\textwidth}
{\codesize

\algotopspace{}
\algoinput Input array $x$, level $\ell$, size $n$, entry $i$, stop at level $\ell'$
\algooutput DFT entry $y^{\ell}[i]$ 
\algomiddlespace

\noindent
\begin{enumerate}
\setlength{\itemindent}{-1.5em}

\vsitem \xif $\ell = \ell'$ \xthen \xreturn $x[i]$ \label{line:fft-sp-basecase}
\vsitem $j \gets i - (n/2) \times [i < n/2]$; $w \gets w_n^{-j}$
\vsitem \xpar $u \gets$ \textsc{FFT-sp}$(x,\ell + 1,n/2,j, \ell')$
\vsitem[] \xblankpar $v \gets $ \textsc{FFT-sp}$(x + 2^{\ell},\ell + 1,n/2,j, \ell')$
\vsitem \xreturn $u - wv \times (-1)^{[i < n/2]}$

\algobottomspace{}
\end{enumerate}
}
\end{minipage}
\end{mycolorbox}
\end{minipage}
\caption{The 2-way \textsc{FFT} algorithm with stages.}
\label{fig:fft-2way-stages}
\end{minipage}
\end{figure*}

\begin{lemma}
\label{lem:fft-2way-stages}
The 2-way FFT algorithm with stages has a complexity of $\Oh{(1/\psi) n^{1 + \psi}}$ work, $\Oh{n}$ space, and $\Oh{(1/\psi) \log n}$ span, where $\psi \in [1/\log n, 1]$. 
\end{lemma}

\begin{proof} $[$Work.$]$ There are $1/\psi$ stages. The height of each stage is $\psi \log n$. Total number of cells at the root nodes of each stage is $\Oh{n}$. The computation of each cell requires $\Oh{2^{\psi \log n}} = \Oh{n^{\psi}}$ work. So, the total work is $\Oh{(1/\psi) n^{1 + \psi}}$. [Span.] The execution of stages is sequential. As there are $\Oh{n}$ cells at the root nodes of each stage, the span for launching these cells is $\Oh{\log n}$. The span for computing a DFT entry (i.e., \textsc{FFT-sp}) is $\Oh{(1/\psi) \log n}$. [Space.] Storing the DFT array at each of the $1/\psi$ stages requires $\Oh{(1/\psi) n}$ space. However, we can reuse two arrays to perform all computations. Hence, we just need $\Oh{n}$ space.
\end{proof}

\paragraph{$\sqrt{n}$-way FFT and 2-way FFT with Stages.} The core idea of the algorithm is as follows. We execute the $\sqrt{n}$-way FFT algorithm for the first $\log (1/\epsilon)$ levels of the recursion tree, where $1/\epsilon \in [2, \log n]$. We then switch to the 2-way FFT algorithm with stages.

\begin{theorem} 
\label{thm:fft-rootnway-2way}
The $\sqrt{n}$-way combined with 2-way with stages FFT algorithm has a complexity of $\Oh{n \log^{f(n, \epsilon)} n}$ work, $\Oh{n}$ space, and $\Oh{\log(1/\epsilon) \log n}$ span, where $1/\epsilon \in [2, \log n]$ and $f(n, \epsilon) = (1/\log \log n) (\log n / ((1/\epsilon) \log (1/\epsilon) ) + \log ((1/\epsilon) \log (1/\epsilon) ) )$. 
\end{theorem}

\begin{proof}
	From Lemma \ref{lem:fft-2way-stages}, we have the following bounds for the 2-way FFT: $T_{1}'(m,\psi) = \Oh{ (1/\psi) m^{1+ \psi} }$ and $T_{\infty}'(m,\psi) = \Oh{ (1/\psi) \log m }$, where, $\psi \in [1/\log m, 1]$. Then, the work and span recurrences for the $\sqrt{n}$-way combined with 2-way with stages algorithm are as follows:
	\begin{align*}
	&T_1(m) \le \begin{cases}
	T_1'(m, \psi) & \text{if } m \le n^{\epsilon}, \\
	2m^{1/2} T_1(m^{1/2}) + \Oh{m} & \mbox{if } m > n^{\epsilon}. \end{cases} &&T_{\infty}(m) \le \begin{cases}
	T_\infty'(m, \psi) & \text{if } m \le n^{\epsilon}, \\
	2 T_\infty(m^{1/2}) + \Oh{\log m} & \mbox{if } m > n^{\epsilon}. \end{cases}
	\end{align*}
	Expanding the recurrences, we get
	\begin{align*}
	&T_1(n) \le (1/\epsilon) n^{1-\epsilon} T_1'(n^{\epsilon}, \psi) + c\cdot n/\epsilon\\
	&T_{\infty}(n) \le (1/\epsilon) T_{\infty}'(n^{\epsilon}, \psi) + c' \cdot \log (1/\epsilon) \log n
	\end{align*}
	for positive constants $c$ and $c'$. Substituting the values of $T_1'$ and $T_{\infty}'$, we obtain $T_1(n) = \Oh{(1/(\epsilon \psi)) n^{1+\epsilon \psi}}$ and $T_{\infty}(n) = \Oh{(1/\psi + \log (1/\epsilon)) \log n}$. We set $\psi = 1/\log (1/\epsilon)$. Writing our work bound in the form $T_1(n) = \Oh{n \log^{f(n, \epsilon)} n}$, we can easily find a corresponding function $f(n, \epsilon)$. Direct computation shows that $f(n, \epsilon) = (1/\log \log n) (\log n / ((1/\epsilon) \log (1/\epsilon) ) + \log ((1/\epsilon) \log (1/\epsilon) ) )$.
	\begin{align*}
	&n^{1 + \varepsilon \psi} / (\varepsilon \psi) = n \log^{f(n, \epsilon)} n \implies    1 + \varepsilon \psi \log n + \log (1/(\varepsilon \psi)) = 1 + f(n, \epsilon) \log \log n \text{ (taking log)}\\
	&f(n, \epsilon) = ( \varepsilon \psi \log n + \log (1/(\varepsilon \psi)) ) / \log \log n = (1/\log \log n) (\log n / ((1/\epsilon) \log (1/\epsilon) ) + \log ((1/\epsilon) \log (1/\epsilon) ) )
	\end{align*}
	When $\psi = 1/\log (1/\epsilon)$, the span is $\Oh{\log (1/\epsilon) \log n}$. Hence, the theorem follows.
\end{proof}

\hide{ = \frac{ \frac{1}{\frac{1}{\varepsilon}\log \frac{1}{\varepsilon}} \log n + \log (\frac{1}{\varepsilon} \log \frac{1}{\varepsilon})}{\log \log n}}

\begin{corollary}
\label{cor:fft}
The $\sqrt{n}$-way combined with 2-way with stages FFT algorithm has a complexity of $\Oh{n \log^{g(n)} n}$ work, $\Oh{n}$ space, and $\Oh{\log n \log \log \log n}$ span, where $g(n) < 2$ for $n < 10^{10,000}$.
\hide{Setting $\epsilon = \log^{-2} \log n$ in the \textsc{FFT-$\sqrt{n}$:2} algorithm gives a complexity of $\Oh{\log n \log \log \log n}$ span, $\Oh{?}$ space, and $\Oh{ n \log^{2} n}$ work, for all values of $n < 10^{10^{4}}$, where the hidden constants are the same as those in Theorem \ref{thm:fft-rootnway-2way}.}
\end{corollary}
\begin{proof}
Set $1/\epsilon = (\log \log n)^2$ and $g(n) = f(n, \epsilon)$ in Theorem \ref{thm:fft-rootnway-2way}. \hide{Plugging in $\epsilon = \log^{-2} \log n$ gives $f(n, \log^{-2} \log n) = (1/\log \log n) (\log n / ((\log^2 \log n) \log (\log^2 \log n) ) + \log ((\log^2 \log n) \log (\log^2 \log n) ) )$. That this function satisfies the stated inequalities can be seen from the plot shown in figure $\ref{fig:fft-f-plots}$.}
\end{proof}

\section{Conclusion}
\label{sec:conclusion}

In this paper, we presented several fundamental low-span algorithms in the binary-forking model without using locks and atomic instructions. Our parallel algorithms perform work (almost) the same as that of the serial algorithms from which they are derived. All our results improve known results in the binary-forking model with and without atomics.

We introduced the technique of \highlight{single-point computation in stages} through Strassen's MM and FFT to carefully set a balance between high work-sharing and high span of the given algorithm and low work-sharing and low-span of the single-point computation variant to obtain parallel algorithms with optimal/near-optimal span without work blow-up. This technique can be used to design efficient parallel algorithms for other problems too.

We also presented a randomized sorting algorithm with optimal span and optimal work, both bounds are w.h.p. in the number of elements being sorted.

A few interesting problems (in the binary-forking model without using locks and atomic instructions) that one could aim to solve in the future are as follows: (1) Our parallel Strassen's MM algorithm achieves $\Oh{n^w \log \log n}$ work and optimal $\Oh{\log n}$ span (or $\Oh{n^w}$ work and $\Oh{\log n \log \log \log n}$ span). Design a Strassen's MM algorithm with $\Oh{n^{w}}$ work and optimal span. (2) Our randomized comparison sorting algorithm uses concurrent writes and achieves optimal $\Oh{n \log n}$ work and optimal $\Oh{\log n}$ span simultaneously, both bounds are with high probability in $n$, and uses $\omega(n)$ space. Design a randomized sorting algorithm that uses exclusive writes and achieves optimal work and optimal span bounds, both bounds w.h.p., and uses linear space. (3) Design a comparison sorting algorithm that uses exclusive writes and achieves optimal work and optimal span. (4) Our FFT algorithm achieves $\Oh{n \log^{g(n)} n}$ work and $\Oh{\log n \log \log \log n}$ span, where $g(n) < 2$ for $n < 10^{10,000}$. Design an FFT algorithm with $\Oh{n \log n}$ work and $\Oh{\log n}$ span.
\subsection*{Acknowledgments}
This research was supported by NSF grants CNS-1553510, CCF-1439084, CNS-1938709, CCF-1617618, CCF-1716252, CCF-1725543,  and CCF-1725428.
\clearpage
\bibliography{references}

\clearpage
\section{Appendix}

\subsection{$n^{\epsilon}$-way Merge Sort} \label{ssec:epsilon-sort}
In this section, we present a simple merge-based parallel sorting algorithm parameterized on a fixed constant $\epsilon \in (0, 1]$, which achieves the span of $\Oh{(\epsilon + 1/\epsilon) \log n}$ but performs suboptimal work of $\Oh{(1/\epsilon) n^{1 + \epsilon}}$. We call this algorithm {\EpsilonWaySort}. 

\begin{figure*}[h]
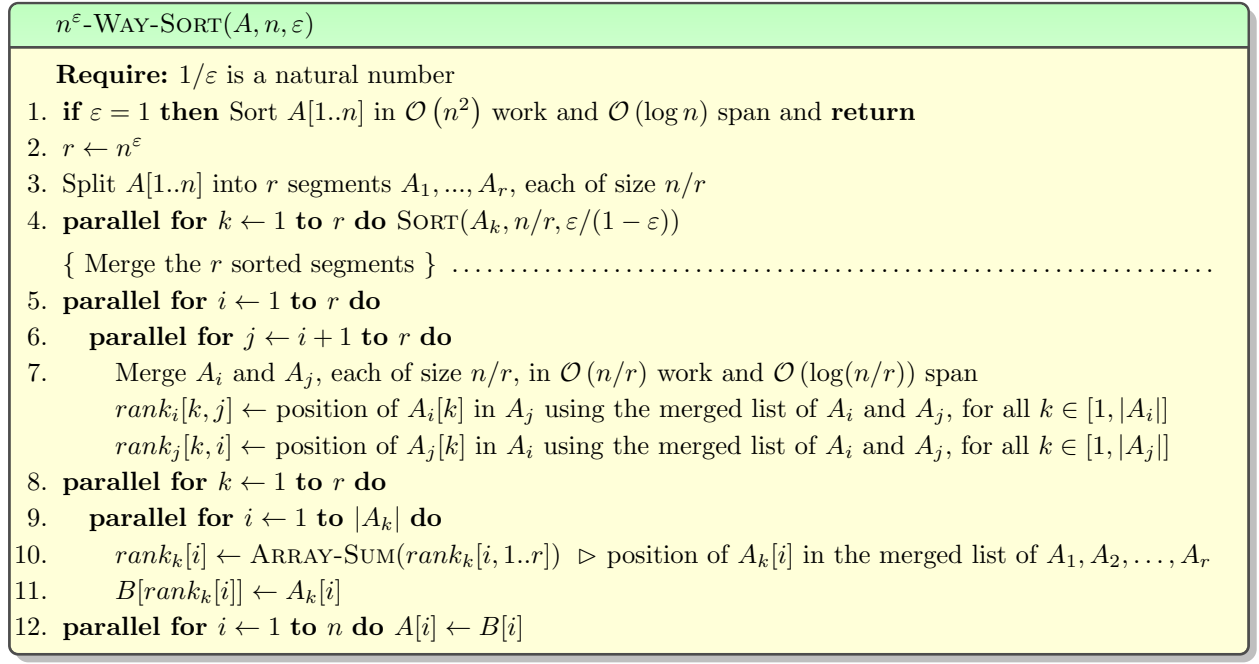

\centering
\begin{minipage}{\textwidth}
\begin{mycolorbox}{$\EpsilonWaySort(A, n, \epsilon)$}
\begin{minipage}{0.99\textwidth}
{\codesize
\algotopspace{}
\vspace{0.1cm}
\algorequire $1/\epsilon$ is a natural number
\noindent
\begin{enumerate}
\setlength{\itemindent}{-1.5em}
\vsitem \xif $\epsilon = 1$ \xthen Sort $A[1..n]$ in $\Oh{n^2}$ work and  $\Oh{\log n}$ span and \xreturn
\vsitem $r \gets n^\varepsilon$
\vsitem Split $A[1..n]$ into $r$ segments $A_1, ..., A_r$, each of size $n/r$
\vsitem \xparallelfor $k \gets 1$ \xto $r$ \xdo $\textsc{Sort}(A_k, n/r, \varepsilon / (1 - \varepsilon))$
\algomiddlespace{}
\vsitem[] $\{ \text{ Merge the $r$ sorted segments } \}$ \dotfill
\vsitem \xparallelfor $i \gets 1$ \xto $r$ \xdo
\vsitem \T \xparallelfor $j \gets i + 1$ \xto $r$ \xdo
\vsitem \T \T Merge $A_i$ and $A_j$, each of size $n/r$, in $\Oh{n/r}$ work and $\Oh{\log (n/r)}$ span
\vsitem[] \T \T $rank_i[k, j] \gets$ position of $A_i[k]$ in $A_j$ using the merged list of $A_i$ and $A_j$, for all $k \in [1, |A_i|]$
\vsitem[] \T \T $rank_j[k, i] \gets$ position of $A_j[k]$ in $A_i$ using the merged list of $A_i$ and $A_j$, for all $k \in [1, |A_j|]$
\vsitem \xparallelfor $k \gets 1$ \xto $r$ \xdo
\vsitem \T \xparallelfor $i \gets 1$ \xto $|A_k|$ \xdo
\vsitem \T \T $rank_k[i] \gets \textsc{Array-Sum}(rank_k[i, 1..r])$ \xcomment{position of $A_k[i]$ in the merged list of $A_1,A_2,\ldots, A_r$}
\vsitem \T \T $B[rank_k[i]] \gets A_k[i]$
\vsitem \xparallelfor $i \gets 1$ \xto $n$ \xdo $A[i] \gets B[i]$

\algobottomspace{}
\end{enumerate}
}
\end{minipage}
\end{mycolorbox}
\end{minipage}
\vspace{-0.4cm}
\caption{The $n^{\epsilon}$-way sorting algorithm.}
\label{fig:multiway-sort}
\vspace{-0.2cm}
\end{figure*}

Figure \ref{fig:multiway-sort} gives a pseudocode of the sorting algorithm. Choose a fixed constant $\epsilon \in (0, 1]$ such that $1/\epsilon$ is a natural number and suppose that $r = n^{\epsilon}$. The input array $A[1..n]$ is split into $r$ subarrays, each having $n/r$ elements. All $r$ subarrays are sorted recursively. These $r$ sorted subarrays are then merged. The merging process consists of two stages.

In the first stage, we compute $rank_i[k, j]$ in parallel, for $i \in [1, r]$, $j \in [i + 1, r]$, and $k \in [1, |A_i|]$, where $rank_i[k, j]$ represents the rank or position of $A_i[k]$ (the $k$th element of subarray $A_i$) in $A_j$ during the process of merging. Similarly, we compute $rank_j[k, i]$ in parallel. We do not need to store the merged list. We simply need to find the ranks of different elements of the array segments in $\Oh{n}$ time and $\Oh{\log n}$ span, which is easy to achieve.

In the second stage, we compute $rank_k[i]$ in parallel, for $k \in [1, r]$ and $i \in [1, |A_k|]$, where $rank_k[i]$ represents the position of $A_k[i]$ (the $k$th element of subarray $A_i$) in the merged list of the $r$ subarrays $A_1,A_2,\ldots, A_r$. We then use these values to sort $A$ by corresponding assignments. 

\begin{lemma}
{\EpsilonWaySort} has a complexity of $\Oh{(1/\epsilon) n^{1 + \epsilon}}$ work and $\Oh{(\epsilon + 1/\epsilon) \log n}$ span, for $\epsilon \in (0, 1]$.
\end{lemma}
\begin{proof}
The work and span recurrences for the algorithm are:
\begin{align*}
&T_1(n, \varepsilon) \leq \begin{cases}
    c_1n^2 & \text{if } \varepsilon = 1, \\
    n^\varepsilon T_1(n^{1-\varepsilon}, \frac{\varepsilon}{1 - \varepsilon}) + c_2 n^{1 + \varepsilon} & \text{if } \varepsilon \ne 1. \end{cases} 
    &&T_\infty(n, \varepsilon) &\leq \begin{cases}
    c_1' \log n & \text{if } \varepsilon = 1, \\
    T_\infty(n^{1 - \varepsilon}, \frac{\varepsilon}{1 - \varepsilon}) + c_2' \log n & \text{if } \varepsilon \ne 1. \end{cases}
\end{align*}
Expanding the recurrences and setting $k = 1/\varepsilon - 1$, we get
\begin{align*}
T_1(n, \varepsilon) &\leq n^\varepsilon \left[n^\varepsilon T_1(n^{1 - 2\varepsilon}, \frac{\varepsilon}{1 - 2\varepsilon}) + c_2 n\right] + c_2 n^{1 + \varepsilon} = n^{2\varepsilon} T_1(n^{1 - 2\varepsilon}, \frac{\varepsilon}{1 - 2\varepsilon}) + 2 c_2 n^{1 + \varepsilon} \\
    &\leq n^{2\varepsilon} \left[ n^\varepsilon T_1(n^{1 - 3\varepsilon}, \frac{\varepsilon}{1 - 3\varepsilon}) + c_2 n^{1 - \varepsilon}\right] + 2 c_2 n^{1 + \varepsilon} = n^{3\varepsilon} T_1(n^{1 - 3\varepsilon}, \frac{\varepsilon}{1 - 3\varepsilon}) + 3 c_2 n^{1 + \varepsilon} \\
    &\leq n^{k\varepsilon} T_1(n^{1 - k\varepsilon}, \frac{\varepsilon}{1 - k\varepsilon}) +  c_2 k n^{1 + \varepsilon} = n^{1 - \varepsilon} T_1(n^\varepsilon, 1) + c_2 k n^{1 + \varepsilon} \\
    &\leq n^{1 - \varepsilon} (n^\varepsilon)^2 + c_2 k n^{1+\varepsilon} = n^{1 + \varepsilon} + c_2 k n^{1 + \varepsilon} = \Oh{\frac{1}{\varepsilon} n^{1 + \varepsilon}}\\
T_\infty(n, \varepsilon) &\leq T_\infty(n^{1 - k\varepsilon}, \frac{\varepsilon}{1 - k \varepsilon}) + c_2' \log n \left[ 1 + (1 - \varepsilon) + (1 - 2\varepsilon) + \cdots + (1 - (k - 1)\varepsilon)\right] \\
    &= T_\infty(n^{1 - k\varepsilon}, \frac{\varepsilon}{1 - k\varepsilon}) + c_2' \log n \times \frac{k}{2} \left[ 2 - (k - 1)\varepsilon\right] \\
    &= T_\infty(n^\varepsilon, 1) + c_2' (\frac{1}{\varepsilon} + 1 - 2\varepsilon) \log n \leq c_1' \log n^\varepsilon + c_2' (\frac{1}{\varepsilon} + 1) \log n = \Oh{(\varepsilon + \frac{1}{\varepsilon})\log n}
\end{align*}
\end{proof}

\subsection{$d$-D FFT}

In this section, we develop a straightforward generalization of the $\sqrt{n}$-way FFT used in theorem \ref{thm:fft-rootnway-2way}. Instead of splitting up an FFT $y[i] = \sum_{j = 0}^{n-1} a[j] w_n^{-ij}$ into two nested sums $y[i] = \sum_{j_1 = 0}^{\sqrt{n}-1} \sum_{j_2 = 0}^{\sqrt{n}-1} a[j_1 + \sqrt{n} j_2] w_n^{-i(j_1 + \sqrt{n} j_2)}$, we will break it into $d$ nested sums $$y[i] = \sum_{j_1 = 0}^{\sqrt[d]{n}-1} \cdots \sum_{j_d = 0}^{\sqrt[d]{n}-1} a\left[\sum_{k = 1}^{d} j_k n^{(k-1)/d}\right] w_n^{-i\sum_{k = 1}^{d} j_k n^{(k-1)/d}}.$$

\begin{lemma}
The \textsc{$d$-D-FFT} algorithm takes $\Oh{d n \log n}$ work and $\Oh{d \log n \log_d \log n}$ span to run on an array of size $n$.
\end{lemma}
\begin{proof}
The code given in Figure \ref{fig:d-D-fft} shows the structure of the algorithm being used. We reindex the array as a column-major $d$-D hypercube with side length $\sqrt[d]{n}$; for every dimension we recursively apply \textsc{$d$-D-FFT} across all $n^{1 - 1/d}$ subarrays found by holding all but one of the $d$ indices constant, followed by a single pointwise multiplication by twiddle factors. The recurrences from this for work and span are
$$T_1(n) = \begin{cases} \Oh{1} & \text{if } n \leq 1 \\ d \; n^{1 - 1/d} T_1(n^{1/d}) + \Oh{d n} & \text{if } n > 1 \end{cases},$$
$$T_\infty(n) = \begin{cases} \Oh{1} & \text{if } n \leq 1 \\ d \; T_\infty(n^{1/d}) + \Oh{d \log n} & \text{if } n > 1 \end{cases},$$
the solutions to these are given by $T_\infty(n) = \Oh{d \log n \log_d \log n}$ and $T_1(n) = \Oh{d n \log n}$.
\end{proof}

\begin{figure*}
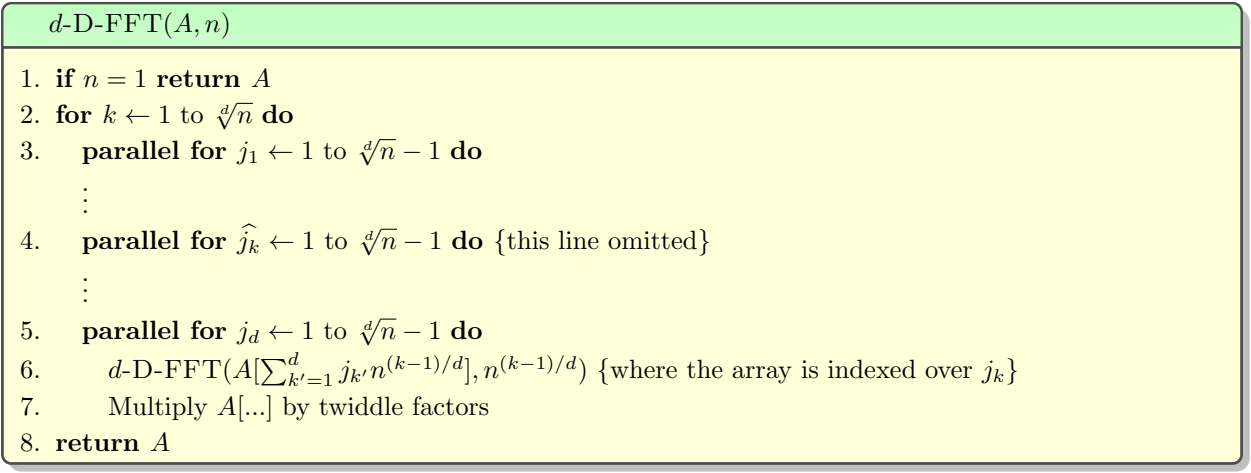

\centering
\begin{minipage}{\textwidth}
\begin{mycolorbox}{\textsc{$d$-D-FFT}$(A, n)$}
\begin{minipage}{0.99\textwidth}
{\codesize
\algotopspace{}
\vspace{0.1cm}
\noindent
\begin{enumerate}
\setlength{\itemindent}{-1.5em}
\vsitem \xif $n = 1$ \xreturn $A$
\vsitem \xfor $k \gets 1$ to $\sqrt[d]{n}$ \xdo
\vsitem \T \xparallelfor $j_1 \gets 1$ to $\sqrt[d]{n} - 1$ \xdo
\vsitem[] \T \vdots
\vsitem \T \xparallelfor $\widehat{j_k} \gets 1$ to $\sqrt[d]{n} - 1$ \xdo \{this line omitted\}
\vsitem[] \T \vdots
\vsitem \T \xparallelfor $j_d \gets 1$ to $\sqrt[d]{n} - 1$ \xdo
\vsitem \T \T \textsc{$d$-D-FFT}$(A[\sum_{k' = 1}^{d} j_{k'} n^{(k-1)/d}], n^{(k - 1) / d})$ \{where the array is indexed over $j_k$\}
\vsitem \T \T Multiply $A[...]$ by twiddle factors
\vsitem \xreturn $A$

\algobottomspace{}
\end{enumerate}
}
\end{minipage}
\end{mycolorbox}
\end{minipage}
\vspace{-0.4cm}
\caption{The $d$-D FFT algorithm.}
\label{fig:d-D-fft}
\vspace{-0.2cm}
\end{figure*}

\begin{figure*}
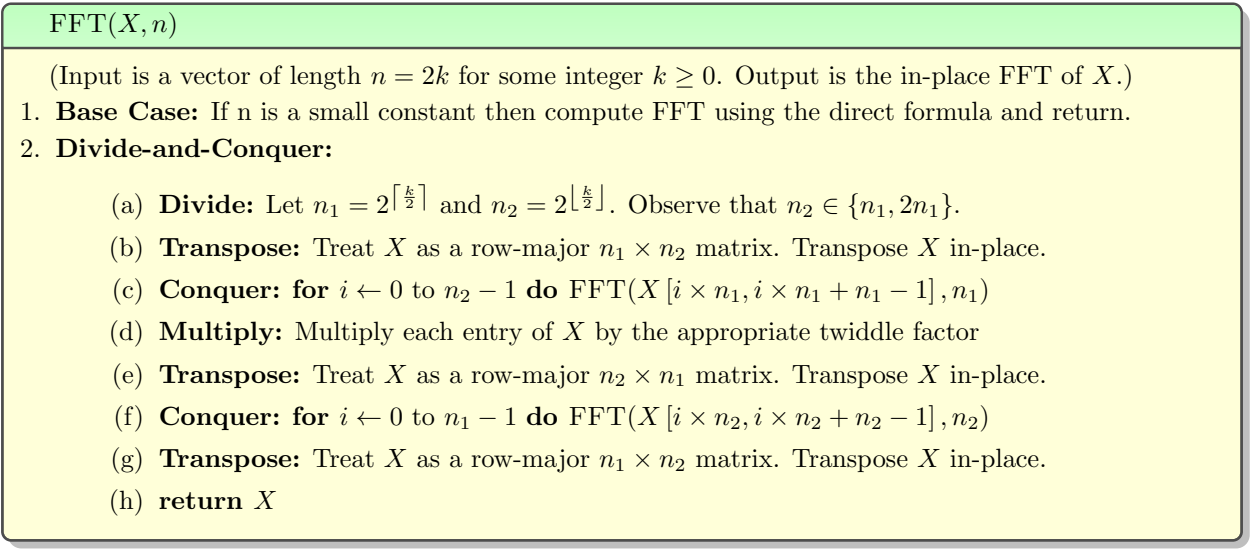
 % transcription of The Fat Fourier Transform (from homework)
\centering
\begin{minipage}{\textwidth}
\begin{mycolorbox}{$\textsc{FFT}(X, n)$}
\begin{minipage}{0.99\textwidth}
{\codesize
\algotopspace{}
\vspace{0.1cm}
(Input is a vector of length $n = 2k$ for some integer $k \geq 0$. Output is the in-place FFT of $X$.)
\noindent
\begin{enumerate}
\setlength{\itemindent}{-1.5em}
\vsitem \textbf{Base Case:} If n is a small constant then compute FFT using the direct formula and return.
\vsitem \textbf{Divide-and-Conquer:}
\begin{enumerate}
    \item \textbf{Divide:} Let $n_{1}=2^{\left\lceil\frac{k}{2}\right\rceil}$ and $n_{2}=2^{\left\lfloor\frac{k}{2}\right\rfloor}$. Observe that $n_{2} \in\left\{n_{1}, 2 n_{1}\right\}$.
    \item \textbf{Transpose:} Treat $X$ as a row-major $n_1 \times n_2$ matrix. Transpose $X$ in-place.
    \item \textbf{Conquer:} \xfor $i\gets 0$ to $n_2 - 1$ \xdo $\textsc{FFT}(X\left[i \times n_{1}, i \times n_{1}+n_{1}-1\right], n_{1})$
    \item \textbf{Multiply:} Multiply each entry of $X$ by the appropriate twiddle factor
    \item \textbf{Transpose:} Treat $X$ as a row-major $n_2 \times n_1$ matrix. Transpose $X$ in-place.
    \item \textbf{Conquer:} \xfor $i\gets 0$ to $n_1 - 1$ \xdo $\textsc{FFT}(X\left[i \times n_{2}, i \times n_{2}+n_{2}-1\right], n_{2})$
    \item \textbf{Transpose:} Treat $X$ as a row-major $n_1 \times n_2$ matrix. Transpose $X$ in-place.
    \item \xreturn $X$
\end{enumerate}

\algobottomspace{}
\end{enumerate}
}
\end{minipage}
\end{mycolorbox}
\end{minipage}
\vspace{-0.4cm}
\caption{The generic \textsc{FFT} algorithm.}
\label{fig:fft-multiway}
\vspace{-0.2cm}
\end{figure*}

\begin{figure*}
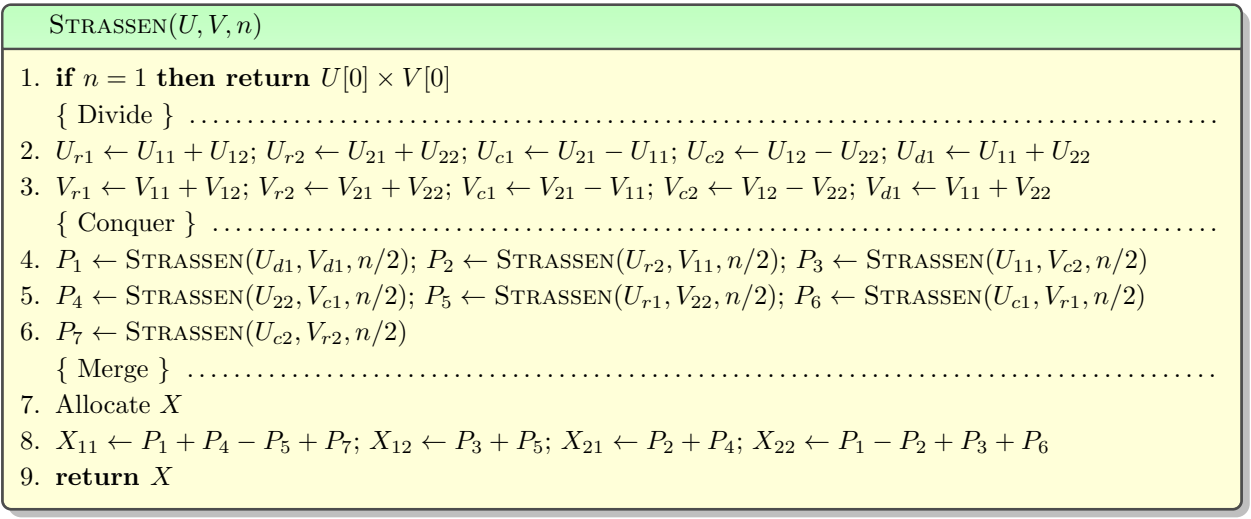

\centering
\begin{minipage}{\textwidth}
\begin{mycolorbox}{\textsc{Strassen}$(U, V, n)$}
\noindent
{\codesize
\algotopspace{} 
\vspace{-0.1cm}
\noindent
\begin{enumerate}
\setlength{\itemindent}{-1.5em}
\vsitem \xif $n = 1$ \xthen \xreturn $U[0] \times V[0]$
\vsitem[] $\{ \text{ Divide } \}$ \dotfill
\vsitem $U_{r1} \gets U_{11} + U_{12}$; $U_{r2} \gets U_{21} + U_{22}$; $U_{c1} \gets U_{21} - U_{11}$; $U_{c2} \gets U_{12} - U_{22}$; $U_{d1} \gets U_{11} + U_{22}$
\vsitem $V_{r1} \gets V_{11} + V_{12}$; $V_{r2} \gets V_{21} + V_{22}$; $V_{c1} \gets V_{21} - V_{11}$; $V_{c2} \gets V_{12} - V_{22}$; $V_{d1} \gets V_{11} + V_{22}$
\vsitem[] $\{ \text{ Conquer } \}$ \dotfill
	\vsitem $P_1 \gets \textsc{Strassen}(U_{d1}, V_{d1}, n/2)$; $P_2 \gets \textsc{Strassen}(U_{r2}, V_{11}, n/2)$; $P_3 \gets \textsc{Strassen}(U_{11}, V_{c2}, n/2)$ 
	\vsitem $P_4 \gets \textsc{Strassen}(U_{22}, V_{c1}, n/2)$; $P_5 \gets \textsc{Strassen}(U_{r1}, V_{22}, n/2)$; $P_6 \gets \textsc{Strassen}(U_{c1}, V_{r1}, n/2)$
	\vsitem $P_7 \gets \textsc{Strassen}(U_{c2}, V_{r2}, n/2)$
	\vsitem[] $\{ \text{ Merge } \}$ \dotfill
	\vsitem Allocate $X$
	\vsitem $X_{11} \gets P_1 + P_4 - P_5 + P_7$; $X_{12} \gets P_3 + P_5$; $X_{21} \gets P_2 + P_4$; $X_{22} \gets P_1 - P_2 + P_3 + P_6$
	\vsitem \xreturn $X$
\end{enumerate}
}
\end{mycolorbox}
\end{minipage}
\vspace{-0.4cm}
\caption{Strassen's MM algorithm.}
\label{fig:strassen-original}
\vspace{-0.2cm}
\end{figure*}

\end{document}